\documentclass[acmtog,nonacm,balance=false]{acmart}
\usepackage{amsmath}
\usepackage{algorithm}
\usepackage{algpseudocode}

\setcopyright{none}
\usepackage{amsthm}
\theoremstyle{remark}
\newtheorem{remark}{Remark}

\newcommand{\wopfigure}[2]{\IfFileExists{#1}{\includegraphics[width=\linewidth,height=#2,keepaspectratio]{#1}}{\fbox{\parbox[c][#2][c]{0.94\linewidth}{\centering\small Figure asset not included\\\texttt{\detokenize{#1}}}}}}
\begin{document}
\title{Walk on Multiscale Media: Compiling Editable and Composable Transport Medium Prefabs for Monte Carlo PDE Solvers}
\author{Yunfeng Liao}
\author{Xiucheng Li}
\affiliation{\institution{Harbin Institute of Technology (Shenzhen)}\city{Shenzhen}\state{Guangdong}\country{China}}
\renewcommand{\shortauthors}{Liao and Li}
\begin{abstract}
Monte Carlo PDE solvers enable mesh-free pointwise queries, but repeatedly traversing heterogeneous media is costly, particularly in scenes with reusable or dynamically changing materials. We introduce \textbf{Walk on Multiscale Media (WoMM)}, which compiles fine-scale transport into reusable, editable and composable \textit{transport prefabs}. A probabilistic interpretation of multiscale finite-element bases and condensed operators yields coarse transitions retaining fine-scale material effects. Each prefab combines port-level operators, entrance-dependent exit distributions, and Poisson-marked trajectory records, enabling incremental updates for overlapping absorption fields, gradient-field steering through path reweighting and transient diffusion. Integrated with conventional Walk-on-Spheres in a hybrid continuous--discrete solver, these prefabs accelerate repeated pointwise PDE queries and support interactive material edits without repeatedly solving the fine-scale problem.
\end{abstract}

\keywords{Monte Carlo PDE Solvers, Walk-on-Spheres, Multiscale Methods, Heterogeneous Media, Stochastic Transport}

\begin{teaserfigure}
    \centering
    \includegraphics[width=0.95\textwidth]{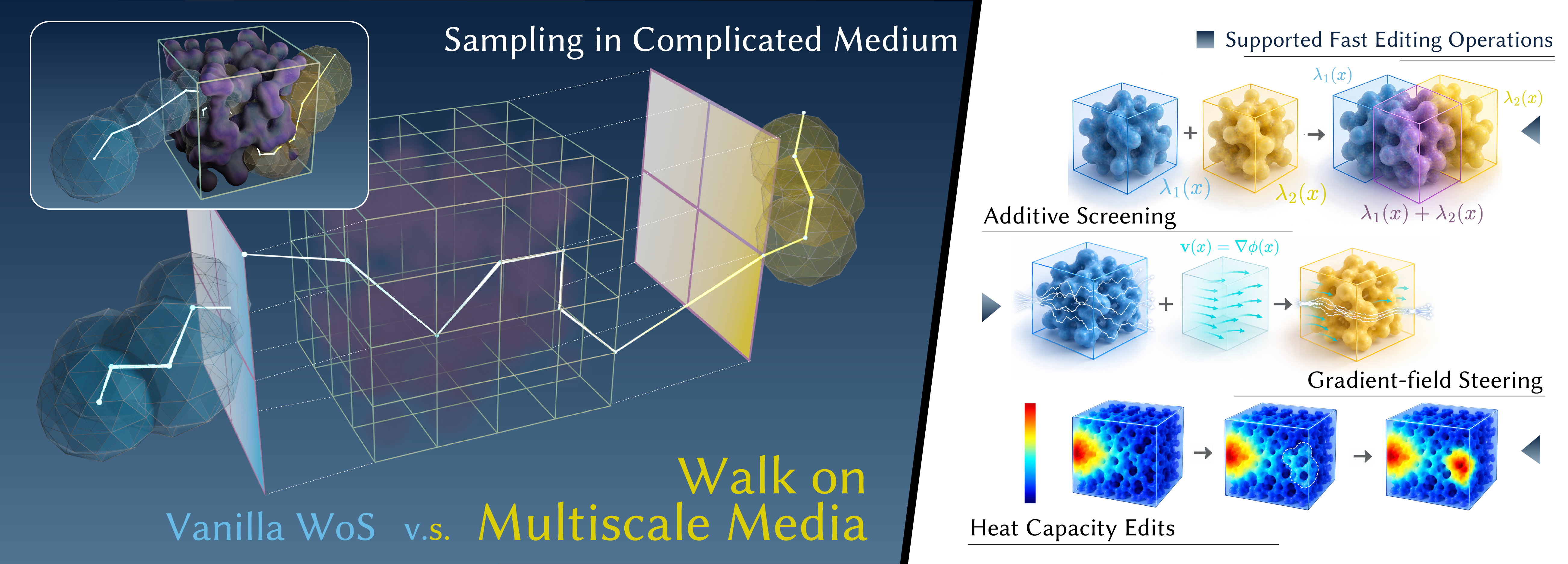}
\caption{Complex heterogeneous media make conventional random walks expensive to evaluate repeatedly. Our method compiles their fine-scale transport into reusable multiscale prefabs, allowing stochastic trajectories to traverse intricate materials through efficient macro-transitions. The same precomputed transport representation supports interactive physical edits through cached-path reweighting, without reconstructing the underlying heterogeneous medium.}
    \Description{1}
    \label{fig:teaser}
\end{teaserfigure}
\maketitle
\thispagestyle{fancy}\fancyfoot[C]
{\footnotesize Preprint}

\section{Introduction}
\label{sec:introduction}
Heterogeneous media are ubiquitous in computer graphics, from porous materials and composites to spatially varying absorbers and heat-conducting objects. Their fine-scale structures can strongly influence macroscopic transport, making direct simulation expensive even when only a coarse response is required. This challenge becomes particularly relevant in interactive scenes, where material instances are repeatedly queried, spatially composed, or modified. Ideally, the transport behavior of a complex medium should be compiled once into a reusable representation, much like geometric assets, while remaining responsive to changes in the surrounding environment.

Monte Carlo PDE solvers, particularly Walk on Spheres (WoS), enable pointwise solution evaluation without constructing a global volumetric discretization~\cite{muller1956continuous,sawhney2020montecarlo}. Extensions to variable coefficients broaden their applicability to heterogeneous materials~\cite{sawhney2022variable}, but repeated queries still incur substantial stochastic transport costs. In contrast, multiscale finite-element methods (MsFEM) compress fine-scale heterogeneity into effective operators~\cite{hou1997msfem,e2003hmm,efendiev2013gmsfem}, although these operators are typically designed for coarse discretized systems rather than reusable pointwise stochastic queries. Recent approaches, including Boundary Value Caching~\cite{miller2023boundary}, Walk on Probes~\cite{huang2026probes}, Walk on Decomposed Subdomains~\cite{jambon2026decomposed}, and Monte Carlo Steklov Operators~\cite{maesumi2026steklov}, further demonstrate the potential of reusing stochastic computations. However, combining multiscale heterogeneous transport with arbitrary entrance queries, incremental material editing, and spatial composition remains challenging.

We introduce \emph{Walk on Multiscale Media}, a framework that compiles heterogeneous regions into reusable \emph{transport prefabs}. Our key insight is that multiscale basis functions can be interpreted as first-exit probabilities, connecting operator condensation with stochastic transport. Each prefab combines a compact port-level operator, entrance-dependent exit distributions, and precomputed random-walk events. Rather than recomputing fine-scale transport after every material modification, these cached events support incremental updates through path reweighting. The resulting prefabs can be instantiated, composed, and integrated into continuous Monte Carlo solvers, while adaptive resolution controls the cost of evaluating their transport responses.

Our main contributions are:
\begin{itemize}
    \item \textbf{Reusable multiscale transport prefabs.} We establish a probabilistic formulation of multiscale basis functions and condensed operators, compiling fine-scale heterogeneous transport into coarse port responses, entrance-dependent exit distributions, and reusable stochastic path caches.

    \item \textbf{Editable and composable stochastic transport.} We develop cached-event reweighting and an edit-local resolvent estimator for incremental material updates without resampling spatial trajectories. The same representation supports overlapping absorption fields, gradient-field steering, and temporal transport edits.

    \item \textbf{Adaptive continuous--discrete Monte Carlo simulation.} We integrate transport prefabs with exterior WoS through position-resolving macro-transitions and introduce query-centric multiscale selection to balance transport accuracy against computational cost.
\end{itemize}

\section{Related Work}
\paragraph{Homogenization and multiscale discretization.}
Rapidly oscillating diffusivities arise in composites, porous media, and heterogeneous heat conduction, where unresolved channels and barriers can affect macroscopic fluxes. Classical homogenization characterizes these effects through effective operators rather than pointwise coefficient averaging~\cite{bensoussan1978asymptotic,allaire1992homogenization,jikov1994homogenization}. Numerical approaches avoid explicitly resolving every microscopic feature at query time: MsFEM incorporates heterogeneous coefficients into local elliptic basis functions~\cite{hou1997msfem}, while heterogeneous multiscale methods couple microscopic computations to macroscopic models~\cite{e2003hmm}. GMsFEM enriches local approximation spaces~\cite{efendiev2013gmsfem}; localized orthogonal decomposition (LOD) and operator-adapted multiresolution methods construct energy-adapted bases for rough coefficients~\cite{malqvist2014localization,owhadi2017multigrid}. These methods motivate our treatment of fine-scale transport through local operator compression. Our objective, however, is an editable stochastic interface rather than only an accelerated coarse-grid solve.

\paragraph{Walk-on-spheres for heterogeneous and transient PDEs.}
Walk on Spheres (WoS) estimates elliptic solutions from Brownian first-exit distributions~\cite{muller1956continuous}; its modern graphics formulations enable mesh-free, pointwise PDE evaluation on complex domains~\cite{sawhney2020montecarlo}. Sawhney et al.~\cite{sawhney2022variable} extend WoS to spatially varying diffusion and reaction coefficients using transformations and null-scattering estimators, while Walk on Stars supports mixed Neumann--Dirichlet conditions~\cite{sawhney2023stars}. Recent work further addresses general divergence-form operators via precomputed inverse-Laplacian and harmonic-extension operators~\cite{cimpean2026extended}, and transient heat propagation through space--time walks~\cite{bao2026heatstars,zhou2026timedependent}. These approaches broaden the PDEs accessible to stochastic solvers while this work turns to amortizing repeated transport queries and local material edits in a fixed, highly oscillatory medium.

\paragraph{Reusable Stochastic Operators and Compositional Simulation.}
Boundary Value Caching reuses estimated boundary solutions~\cite{miller2023boundary}, while Walk on Probes shares stochastic estimates through fixed spherical probes~\cite{huang2026probes}. Walk on Decomposed Subdomains (WoDS) precomputes local first-passage operators and assembles a global interface system, supporting localized recomputation under geometric edits~\cite{jambon2026decomposed}. Monte Carlo Steklov Operators similarly constructs reusable boundary operators through stochastic estimation~\cite{maesumi2026steklov}. Building on these approaches, we employ multiscale transport prefabs that reuse local stochastic responses across queries and material edits.

\section{Preliminaries}
\label{sec:preliminaries}

\subsection{Multiscale Discretization of Elliptic PDEs}
\label{sec:multiscale-background}

Highly heterogeneous media contain conducting channels and barriers at scales much smaller than the desired simulation resolution. Replacing their coefficients by cell averages can discard the resulting long-range transport effects. MsFEM incorporates fine-scale heterogeneity into coarse representations by constructing multiscale basis functions from local elliptic problems, specifically satisfying on a subdomain $K$
\begin{equation}
\label{eq:prefab-local-pde}
-\nabla\!\cdot\!\bigl(a_K(\mathbf x)\nabla u\bigr)
+\lambda_K^0(\mathbf x)u=0\quad\text{in }K,
\quad u|_{\partial K}=g,
\end{equation}
with  heterogeneous diffusivity $a_K>0$ and screening coefficient $\lambda_K^0\geq0$. The local problem admits the energy functional $E_K(u)=\frac12\int_K(a_K|\nabla u|^2+\lambda_K^0u^2)\,d\mathbf x$, whose minimizer under prescribed Dirichlet data gives the screened harmonic extension as MsFEM basis functions. To construct these basis functions probabilistically, one can discretize the same local energy on a fine-grid graph using symmetric nonnegative conductances $c_{xy}=c_{yx}\geq0$ between neighboring nodes and lumped screening coefficients $\lambda_K^0(x)\geq0$, giving the resulting discrete energy $E_h(\mathbf u)=\frac12\sum_{\{x,y\}}c_{xy}(u_x-u_y)^2+\frac12\sum_x\lambda_K^0(x)u_x^2=\frac12\mathbf u^\top H_K\mathbf u$. Consequently, the local stiffness matrix satisfies $(H_K)_{xy}=-c_{xy}$ for neighboring $x\ne y$ and $(H_K)_{xx}=\sum_{z\sim x}c_{xz}+\lambda_K^0(x)$ where $z$ is adjacent to $x$. The nonnegative conductances preserve the monotonicity required for the random-walk interpretation developed in Sec.~\ref{sec:multiscale-probability}. Partitioning the fine-grid nodes into $n_b$ boundary and $n_i$ interior nodes yields $H_K=\begin{pmatrix}H_{bb}&-C\\-C^\top&H_{ii}\end{pmatrix}$, where $C\geq0$ contains the boundary--interior couplings and $H_{ii}$ is a nonsingular $M$-matrix. To define the coarse degrees of freedom, we partition the fine boundary into $p$ disjoint ports and introduce $T\in\{0,1\}^{n_b\times p}$, which assigns each boundary node to exactly one port ($T\mathbf1_p=\mathbf1_b$). 

For port values $\mathbf v \in \mathbb R^p$, the boundary trace is $\mathbf u_b=T\mathbf v$. The interior equilibrium equations $H_{ii}\mathbf u_i=C^\top\mathbf u_b$ follow from the homogeneous discrete PDE. Since $H_{ii}$ is the interior Dirichlet operator, its inverse $G_K=H_{ii}^{-1}$ is the discrete Green operator, mapping interior source terms to solutions with zero boundary values. Consequently, defining $R=C^\top T$ and $\Phi_K=G_KR$ gives the screened harmonic extension $\mathbf u_i=\Phi_K\mathbf v$ of the prescribed port values. Thus the corresponding fine-grid state is $\mathbf u=(\mathbf u_b,\mathbf u_i)^\top$. Minimizing the energy $\min_{\mathbf u_i} \mathbf u^\top H_K \mathbf u$ by eliminating $\mathbf u_i$ yields the minimum $\tfrac12 \mathbf v^\top S_K \mathbf v$ where \begin{equation}
    S_K={T^\top H_{bb}T}-{R^\top H_{ii}^{-1}R}=:D_K-J_K.
\end{equation} 

The Schur complement $S_K$ is the corresponding discrete Dirichlet-to-Neumann operator: $D_K$ accounts for boundary contributions, while $J_K=R^\top\Phi_K$ captures transport through the eliminated interior. Thus the coarse system retains the fine-scale response without resolving its interior at every query.

\subsection{Random-Walk Interpretation of Multiscale Operators}
\label{sec:multiscale-probability}
Let $L_0$ denote the unscreened conductance generator, $(L_0f)(x)=\sum_{y\sim x}c_{xy}[f(y)-f(x)]$. Its continuous-time Markov chain (CTMC) holds at $x$ for $t\sim\mathrm{Exp}(d_x^0)$, where $d_x^0=\sum_{y\sim x}c_{xy}$, before jumping to $y$ with probability $c_{xy}/d_x^0$. Screening $\kappa(x)\geq0$ introduces killing at rate $\kappa(x)$, giving total event rate $d_x^\kappa=d_x^0+\kappa(x)$ and killing probability $\kappa(x)/d_x^\kappa$. Let $\partial K_h$ be the fine boundary nodes partitioned into ports $\Gamma_\beta=\{b\in\partial K_h:T_{b\beta}=1\}$ for $\beta \in[p]$. For an unscreened walk, with its first boundary-hitting time $\tau$ and $Y=X_\tau\in\partial K_h$, the screened first-exit kernel is $F_\kappa(x,b)=\mathbb E_x^0[\exp(-\int_0^\tau\kappa(X_t)\,dt)\mathbf1_{\{Y=b\}}]$ with port-exit probabilities $\phi_\beta^\kappa(x)=\sum_{b\in\Gamma_\beta}F_\kappa(x,b)$. Thus the screened MsFEM basis satisfies $(\Phi_\kappa)_{x\beta}=\phi_\beta^\kappa(x)$, where $\Phi_\kappa$ denotes the extension under screening $\kappa$ and $\Phi_{\lambda_K^0}=\Phi_K$.

This interpretation extends naturally to the condensed operators. Since $J_\kappa=R^\top\Phi_\kappa$, its entries $(J_\kappa)_{\alpha\beta}=\sum_{x\in i}R_{x\alpha}\phi_\beta^\kappa(x)=\rho_\alpha\mathbb P_{X_0\sim q_\alpha}^\kappa\{Y\in\Gamma_\beta\}$ represent unnormalized transport between ports, where $\rho_\alpha=\sum_{x\in i}R_{x\alpha}$ and $q_\alpha(x)=R_{x\alpha}/\rho_\alpha$ for $\rho_\alpha>0$. Accordingly, the Schur complement $S_\kappa=D_\kappa-J_\kappa$, with $D_\kappa=T^\top H_{bb}^\kappa T$, induces coarse transitions $P_{\alpha\beta}^{K,\kappa}=-(S_\kappa)_{\alpha\beta}/s_\alpha$ for $\beta\ne\alpha$ and absorption probability $P_{\alpha\dagger}^{K,\kappa}=(S_\kappa\mathbf1_p)_\alpha/s_\alpha$ where $s_\alpha=(S_\kappa)_{\alpha\alpha}>0$ in the sense of identifying stiffness matrix as transport matrix. While these transitions preserve coarse transport, they discard the precise exit location $b$, motivating the additional fine-boundary kernel $F_\kappa$.

\paragraph{Doob conditioning.} The first-exit probabilities also characterize the CTMC conditioned on a prescribed boundary event. With $h(x)=\phi_\beta^\kappa(x)$, $h(b)=\mathbf1_{\{b\in\Gamma_\beta\}}$ and $(L_0-\kappa)h=0$ in the interior, conditioning on this event gives the Doob $h$-transform with jump rates $c_{xy}^h=c_{xy}h(y)/h(x)$ wherever $h(x)>0$, including jumps to boundary nodes, i.e., a conditioned process can be obtained by merely reweighting. As $h(x)=\sum_{y\sim x}c_{xy}h(y) /d_x^\kappa$, the transition probability $P_{xy}^h= \frac{c_{xy}}{d_{x}^\kappa} \frac{h(y)}{h(x)}$ has $\sum_{y}P_{xy}^h=1$ thus having no killing, and also remains its total departure rate since  $\sum_{y\sim x}c_{xy}^h=d_{x}^\kappa$. 

\paragraph{Scalar $h$-transforms.} For a general positive function $h$, the scalar $h$-transform of the unscreened generator is defined by $L_0^h f=h^{-1}L_0(hf)-fh^{-1}L_0h$. On the conductance graph, this yields $L_0^h f(x)=\sum_{y\sim x}c_{xy}^h\allowbreak [f(y)-f(x)]$. Writing $V_h=L_0h/h$, the transformed path law relative to the reference CTMC satisfies $d\mathbb P_x^h/d\mathbb P_x^0=h(Y)/h(x)\allowbreak \exp(-\int_0^\tau V_h(X_t)\,dt)$ under the suitable stopped-process regularity conditions. In the continuous setting, the diffusion generator $\mathcal L_0f=\nabla\cdot(a_K\nabla f)$ transforms into $\mathcal L_0^hf=\mathcal L_0f+2a_K\nabla\log h\cdot\nabla f$, inducing the gradient drift $\mathbf b_h=2a_K\nabla\log h$. By Girsanov's theorem, the logarithmic likelihood ratio is $\int_0^\tau\sqrt{2a_K}\nabla\log h\cdot dB_t- a_K|\nabla\log h|^2dt$. Applying Ito's formula to $d\log h(X_t)$ converts the stochastic integral into $\log[h(Y)/h(x)]-\int_0^\tau(\mathcal L_0h/h-a_K|\nabla\log h|^2)dt$, yielding the same endpoint-weighted representation with continuous correction $\mathcal L_0h/h$. When $h$ is harmonic with respect to the reference generator, the path-dependent correction vanishes, leaving only the endpoint weight $h(Y)/h(x)$.

\section{Editable and Composable Multiscale Transport Prefabs}

We compile each heterogeneous region $K$ into a reusable stochastic transport prefab
\begin{equation}
\label{eq:prefab-method-data}
\text{Prefab}_K=(K,\mathcal A_K^0,T,\Phi_K,S_K,\mathcal E_K, F_K),
\end{equation}
where $\mathcal A_K^0$ specifies the fixed fine-scale diffusion operator, $\Phi_K$ and $S_K$ are the multiscale operators defined in Sec.~\ref{sec:multiscale-background}, and $\mathcal E_K$ is a bank of compact event records instead of recording complete trajectories. For each fine-grid interior node $x$, we generate $N_x$ independent unscreened walks indexed by $r$. Each walk is compressed into an event record $(Y_{x,r},\Pi_{x,r})$, where $Y_{x,r}\in\partial K_h$ and $\Pi_{x,r}$ is the multiset of visited nodes at Poisson candidate events generated along the walk at a fixed rate $\mu$. Reweighting these records under a screening field $\kappa$ reconstructs the first-exit kernel defined in Sec.~\ref{sec:multiscale-probability}:
\begin{equation}
\widehat F_\kappa(x,b)=\frac{1}{N_x}\sum_{r=1}^{N_x}
\mathbf1_{\{Y_{x,r}=b\}} \prod_{z\in\Pi_{x,r}}\left(1-\frac{\kappa(z)}{\mu}\right).
\end{equation}
The same event bank supports incremental updates of the condensed transport operator, while the reconstructed exit kernels provide reusable macro-transitions for hybrid Monte Carlo simulation.

\subsection{Stochastic Prefab Precomputation}
\label{sec:prefab-precomputation}
We precompute the multiscale responses and entrance-dependent exit distributions using compact event records. For the baseline screening $\lambda_K^0$, the multiscale basis is estimated by aggregating the weighted terminal exits over each port:
\begin{equation}
\label{eq:prefab-stochastic-basis}
(\widehat\Phi_K)_{x\beta}
=\sum_{b\in\Gamma_\beta}\widehat F_{\lambda_K^0}(x,b),
\qquad
\widehat S_K=D_K-R^\top\widehat\Phi_K.
\end{equation}
Since $R$ is nonzero only at interior nodes adjacent to the boundary, assembling $\widehat S_K$ requires sampling only at these nodes, without evaluating basis gradients or normal derivatives.

\paragraph{Entrance-dependent transport.}
To enable spatially resolved macro-transitions, we additionally estimate $\widehat F_\kappa(x_a,\cdot)$ at sparse starting anchors $x_a$ distributed throughout $K$. Interior anchors use weighted terminal-exit records from $\mathcal E_K$, whereas boundary anchors have deterministic exits and require no sampling. For a starting point $x$, we select a neighboring anchor $x_a$ according to nonnegative trilinear weights $\omega_a(x)$ and sample its cached exit distribution. This realizes the interpolated kernel
\begin{equation}
\label{eq:prefab-exit-atlas}
\widehat F_\kappa(x,b)
=\sum_{a\in\mathcal N(x)}
\omega_a(x)\widehat F_\kappa(x_a,b),
\quad
\sum_a\omega_a(x)=1.
\end{equation}
The remaining probability corresponds to absorption. Thus a single macro-transition replaces the fine-scale walk through the prefab interior, while preserving an entrance-dependent fine-boundary exit location. Unlike the coarse stiffness assembly, this construction requires sparse spatial coverage of the interior, neither a distribution at every interior node nor full simulation inward.

\paragraph{Port-guided sampling.}
Under strong screening, direct sampling can produce many trajectories with negligible surviving exit weight. We also investigated the Doob conditioning introduced in Sec.~\ref{sec:multiscale-probability} to preferentially sample successful boundary exits. However, the observed improvement was limited. It can be that conditioning on successful exit does not necessarily improve the sampling of localized transport changes, while the guiding function introduces preprocessing overhead and may become less effective after material edits. This motivates shifting the estimation target from the full transport response to its edit-induced difference. Later we use the resolvent identity to express this difference through the edited region and estimate it from the corresponding event records.

\subsection{Incremental Composition via Marked Paths}
\label{sec:incremental-updates}

Overlapping prefabs require different treatments according to whether they describe the same material or additive physical effects. For the same medium, overlap is only representational: if $\mathcal T_1$ and $\mathcal T_2$ are consistent first-exit operators, then $\mathcal T_1u(x)=\mathcal T_2u(x)=u(x)$ for $x\in K_1\cap K_2$. A weighted split or stochastic choice between these representations therefore preserves the expectation without duplicating the medium, up to discretization and interface-interpolation errors. In contrast, physically superposed screening fields add in the overlap, $\lambda(x)=\lambda_1(x)+\lambda_2(x)$. We focus on such nonnegative edits while keeping the underlying conductances fixed. For brevity, quantities before and after edits are labeled by $\text{pre},\text{post}$ or $0,1$. 

\paragraph{Edit-local transport difference.}
For one prefab, write $\kappa_{0}=\lambda_K^0+v_{0}$ and $\kappa_{1}=\kappa_{0}+v$, with $v_{0},v\geq0$. On the interior nodes, let $A_{0}=H_{ii}+\operatorname{diag}(v_{0})$ and $A_{1}=A_{0}+V$, where $V=\operatorname{diag}(v|_i)$. We seek to update the condensed operator \(S_K\). Since the boundary contribution changes explicitly as \(D_K^{\mathrm{post}}=D_K^{\mathrm{pre}}+T^\top V_bT\), it suffices to estimate the change in \(J_K\). Using $\Phi_t=A_t^{-1}R$ and $J_t=R^\top\Phi_t$, the resolvent identity yields
\begin{equation}
\label{eq:prefab-guided-identity}
\Delta J:=J_{0}-J_{1}
=\Phi_{0}^\top V\Phi_{1},
(\Delta J)_{\alpha\beta}
=\sum_{x\in\mathcal S}v(x)\phi_\alpha^{\mathrm{pre}}(x)\phi_\beta^{\mathrm{post}}(x),
\end{equation}
where $\mathcal S=\{x\in i:v(x)>0\}$. Unlike recomputing the modified $J_{1}$ from the ground up, this expression requires first-exit probabilities only at the \emph{edited interior nodes}.

\paragraph{Unbiased replay from event records.}
The bank $\mathcal E_K$ stores the unscreened terminal-exit and Poisson-mark records defined at the beginning of this section. For local screening-edit updates, records are recompiled by reweighting to fetch new path weight $W_{x,r}^\text{post}$ rather than rerunning simulations. Assuming $\kappa_{1}\leq\mu$, we compute 
\begin{equation}
\label{eq:prefab-replay-weights}
W_{x,r}^t=\prod_{z\in\Pi_{x,r}}\left(1-\frac{\kappa_t(z)}{\mu}\right),
a_{r\alpha}=W_{x,r}^{\mathrm{pre}}\mathbf1_{\{Y_{x,r}\in\Gamma_\alpha\}},
b_{r\beta}=W_{x,r}^{\mathrm{post}}\mathbf1_{\{Y_{x,r}\in\Gamma_\beta\}}.
\end{equation}
The corresponding expectations are $\phi_\alpha^{\mathrm{pre}}(x)$ and $\phi_\beta^{\mathrm{post}}(x)$. Since weights from one record are correlated, their sample means cannot simply be multiplied without bias. Instead, define $A_\alpha(x)=\sum_r a_{r\alpha}$, $B_\beta(x)=\sum_r b_{r\beta}$, and $C_\alpha(x)=\sum_r a_{r\alpha}b_{r\alpha}$. To obtain an unbiased estimator of the product $\phi_\alpha^\text{pre}(x)\phi_\beta^\text{post}(x)$, we adopt a second-order U-statistic that pairs distinct independent samples, thereby eliminating the diagonal bias arising from products of sample means:
\begin{equation}
\label{eq:prefab-replay-estimator}
\widehat{\Delta J}_{\alpha\beta}
=\sum_{x\in\mathcal S}\frac{v(x)}{N(N-1)}
\left[A_\alpha(x)B_\beta(x)-\mathbf1_{\{\alpha=\beta\}}C_\alpha(x)\right].
\end{equation}
The correction removes same-record products; those products vanish for $\alpha\ne\beta$ because each record has a single terminal port. We symmetrize $\widehat{\Delta J}$ and update the condensed operator by
\begin{equation}
\label{eq:prefab-incremental-S}
\widehat S_{\mathrm{post}}=\widehat S_{\mathrm{pre}}
+\widehat{\Delta J}+T^\top V_bT.
\end{equation}
No new spatial walks are needed. For the fine-boundary kernels used by the hybrid solver, the same post-edit weights are applied to records at the sparse interior interpolation anchors to obtain $\widehat F_{\kappa_{\mathrm{post}}}$ through the estimator already defined above while boundary anchors remain deterministic. Thus the coarse response is updated from the edited-node records, whereas macro-transition distributions are refreshed only at their interpolation anchors.

\begin{theorem}[Unbiased edit-local replay]
\label{thm:prefab-replay-variance}
Assume that the unscreened records are independent within and across starting nodes, and $0\leq\kappa_{\mathrm{pre}}\leq\kappa_{\mathrm{post}}\leq\mu$. With $N\geq2$ records per edited node, Eq.~\eqref{eq:prefab-replay-estimator}, before or after symmetrization, satisfies
\begin{equation}
\label{eq:prefab-replay-unbiased-variance}
\mathbb E[\widehat{\Delta J}_{\alpha\beta}]=\Delta J_{\alpha\beta},\quad
\operatorname{Var}(\widehat{\Delta J}_{\alpha\beta})
\leq\frac1N\sum_{x\in\mathcal S}v(x)^2.
\end{equation}
\end{theorem}
\begin{proof}
Poisson marking gives $\mathbb E[a_{r\alpha}]=\phi_\alpha^{\mathrm{pre}}(x)$ and $\mathbb E[b_{r\beta}]=\phi_\beta^{\mathrm{post}}(x)$. The bracket in Eq.~\eqref{eq:prefab-replay-estimator} equals $\sum_{r\ne s}a_{r\alpha}b_{s\beta}$, whose expectation is $N(N-1)\phi_\alpha^{\mathrm{pre}}(x)\phi_\beta^{\mathrm{post}}(x)$. This proves unbiasedness by Eq.~\eqref{eq:prefab-guided-identity}. For each $x$, the distinct-record average is an order-two U-statistic with symmetric kernel $(a_{r\alpha}b_{s\beta}+a_{s\alpha}b_{r\beta})/2\in[0,1]$, and hence variance at most $1/N$. Independence across starting nodes gives the bound and symmetrization preserves it by Cauchy--Schwarz.
\end{proof}

\begin{remark}[Computational Cost of Edit-Local Replay]
For $p$ ports, $N$ records per edited node, and at most $M$ marks per record, updating the coarse operator costs
$\mathcal{O}(|S|[N(1+M)+p^2])$
and requires no new spatial trajectories.
\end{remark}

\begin{remark}[Scalar rescaling of diffusion]
    Uniform rescaling of a prefab's diffusion tensor requires no new spatial trajectories. For any $c>0$, replacing $A_K$ with $cA_K$ while keeping the screening field $\kappa$ unchanged yields $F_{cA_K,\kappa}=F_{A_K,\kappa/c}, S_{cA_K,\kappa}=cS_{A_K,\kappa/c}.$ Thus, diffusion rescaling is handled by reweighting the existing Poisson-marked records with effective screening $\kappa/c$, provided $\kappa/c\leq\mu$. In the absence of screening, exit distributions remain unchanged and the condensed operator is simply scaled by $c$. Different coarse cells may be rescaled independently, with their updated operators reassembled into the transport skeleton.
\end{remark}

\subsection{Low-Dimensional Steering and Temporal Edits}
\label{sec:steering-temporal}

\paragraph{Gradient-field steering.}
The scalar $h$-transform introduced in Sec.~\ref{sec:multiscale-probability} provides a way to modify transport direction while preserving the underlying diffusivity and physical screening. In particular, a positive field $h$ induces the drift $\mathbf b_h=2a_K\nabla\log h$, with the path-dependent correction $V_h=L_0h/h$. Rather than prescribing a full-resolution vector field, we parameterize steering by positive port weights $\mathbf w\in\mathbb R_{>0}^p$ and construct their unscreened harmonic extension:
\begin{equation}
\label{eq:steering-harmonic-field}
h_{\mathbf w}|_b=T\mathbf w,\qquad
h_{\mathbf w}|_i=(H_{ii}-\Lambda_i)^{-1}R\mathbf w,
\quad \Lambda_i=\operatorname{diag}(\lambda_K^0|_i).
\end{equation}
Since $L_0h_{\mathbf w}=0$ in the interior, its scalar correction vanishes. The steered first-exit kernel therefore reduces to endpoint reweighting:
\begin{equation}
\label{eq:steering-harmonic-kernel}
F_\kappa^{\mathbf w}(x,b)
=\frac{w_{\beta(b)}}{h_{\mathbf w}(x)}F_\kappa(x,b),
\end{equation}
where $\beta(b)$ identifies the port containing $b$. This transformation preserves the physical screening field $\kappa$ and requires only $p-1$ independent port-weight ratios. For directional control, the weights may be chosen as $w_\beta=\exp(\theta\mathbf d\cdot(\mathbf p_\beta-\mathbf p_K))$, where $\mathbf d$ specifies the preferred direction, $\theta$ controls its strength, and $\mathbf p_\beta,\mathbf p_K$ denote the port and prefab centers.

More generally, nonharmonic steering can also be evaluated using the existing event records. For a record starting at $x$ and terminating at $Y$, the combined steering and screening weight is
\begin{equation}
\label{eq:steering-marked-weight}
W_{x,r}^{h,\kappa}
=\frac{h(Y_{x,r})}{h(x)}
\prod_{z\in\Pi_{x,r}}
\left(1-\frac{\kappa(z)+V_h(z)}{\mu}\right).
\end{equation}
Poisson averaging reproduces the path likelihood ratio and screening survival factor, provided $\kappa+V_h\leq\mu$. Unlike the harmonic construction, a nonzero $V_h$ introduces additional weight variance, particularly when it takes negative values. We therefore favor harmonic steering, for which the cached exit statistics require only endpoint reweighting.

\paragraph{Temporal Heat Capacity Edits.}
We also support spatially varying heat capacity $m_K(x)>0$ in the transient screened diffusion equation
\begin{equation}
\label{eq:temporal-evolution}
m_K\partial_tu-\nabla\cdot(a_K\nabla u)
+(\lambda_K^0+v)u=0.
\end{equation}
For zero initial data, the Laplace transform with parameter $\sigma\geq0$ reduces this equation to the same elliptic problem with effective screening $\kappa_\sigma=\lambda_K^0+v+\sigma m_K$. Its interior operator is
\begin{equation}
\label{eq:temporal-laplace-system}
A_\sigma=H_{ii}+\operatorname{diag}((v+\sigma m_K)|_i).
\end{equation}
Consequently, temporal effects require no new spatial transport process. The existing Poisson marks yield the Laplace-domain weight
\begin{equation}
\label{eq:temporal-marked-weight}
W_{x,r}^{\sigma}
=\prod_{z\in\Pi_{x,r}}
\left(1-\frac{\kappa_\sigma(z)}{\mu}\right),
\quad 0\leq\kappa_\sigma\leq\mu.
\end{equation}
Its conditional expectation introduces the Laplace factor $e^{-\sigma\int_0^\tau m_K(X_t)\,dt}$. Hence screening, non-harmonic steering part and heat capacity enter through the same effective potential $\kappa_\sigma+V_h$, while drift steering contributes only an endpoint factor. These modifications can therefore be composed by reweighting the existing event records. At a fixed $\sigma$, the edit-local U-statistic in Sec.~\ref{sec:incremental-updates} applies directly with $v$ replaced by $\delta v+\sigma\delta m_K$, without introducing another estimator.

\begin{remark}[Extension to direct time-domain walks]
General transient problems, including nonzero initial conditions, can be addressed using time-budgeted walks such as Walk-on-Heat-Balls (WoHB)~\cite{deaconu2018heatballs}, which jointly sample spatial displacements and elapsed times. Let $s$ denote the operational time of the reference diffusion $X_s$, and define the physical clock $\Theta_m(s)=\int_0^s m_K(X_q)\,dq$ for a time-independent capacity field. A walk queried at time $t$ terminates either upon reaching the spatial boundary or when $\Theta_m(s)=t$, contributing boundary data at the remaining time or the initial condition at the stopping position, respectively.

To make such walks reusable, an extended event bank would retain ordered spatial transitions and operational-time increments, time-stamped Poisson events, and terminal states. For a screening field $\lambda(x,t)$, the marks before the stopping time $s_*=\min\{\tau,\Theta_m^{-1}(t)\}$ give the weight
\begin{equation}
\label{eq:temporal-space-time-reweight}
W_{\lambda,t}
=\prod_{\substack{(z,s)\in\Pi\\s<s_*}}
\left(1-\frac{\lambda(z,t-\Theta_m(s))}{\mu}\right),
\end{equation}
where $\mu$ bounds the screening rate. Changing $m_K$ modifies the physical clock and the stopping event, while changing $\lambda$ only requires reevaluating these weights. If the time budget expires within a heat-ball step, its intermediate position must be sampled from the corresponding conditional heat kernel. Exact replay for spatially heterogeneous capacity or screening also requires within-step occupation information, or sufficiently localized steps with constant coefficients. Thus screening edits admit direct event reweighting, whereas capacity edits additionally require reevaluating the physical clock and potentially the stopping location. Such time-domain reuse extends the present framework but requires richer space--time records than the elliptic implementation. 
\end{remark}

\subsection{Hybrid Continuous--Discrete Monte Carlo Walks}
\label{sec:hybrid-walks}

We couple continuous WoS with multiscale transport within each prefab. A prefab $K$ consists of coarse cells $\{K_j\}$, whose condensed operators $S_{K_j}$ define an internal transport skeleton. The fine-boundary kernels $F_{\kappa,j}$ connect continuous positions to this skeleton, while the coarse operators propagate the walk across cells. Upon reaching the outer boundary of the prefab, the exit location is recovered at fine resolution before returning to continuous WoS as described in Algorithm~\ref{alg:hybrid-walk}.

\paragraph{Continuous--coarse coupling.}
Let $M_K\subset E_{\mathrm{in},K}\Subset K$ denote the heterogeneous support, inner trigger region, and outer exit envelope, respectively, with a homogeneous buffer between $M_K$ and $\partial K$. Outside the trigger regions, WoS proceeds on spheres contained in the homogeneous exterior. Once a walk reaches $x\in E_{\mathrm{in},K}$, we locate its containing coarse cell $K_j$ and sample $\widehat F_{\kappa,j}(x,\cdot)$ using the spatial interpolation described in Sec.~\ref{sec:prefab-precomputation}. The sampled location lies on the boundary $\partial K_j$ of the coarse cell (not necessarily on the outer prefab boundary $\partial K$). Its corresponding port initializes the coarse walk while missing probability mass represents absorption.

\paragraph{Coarse transport and fine-boundary exit.}
The local operators are assembled into an internal skeleton by identifying shared coarse-cell ports, giving $S_{\mathrm{skel}}=\sum_j E_j^\top S_{K_j}E_j$, where $E_j$ maps local ports to skeleton degrees of freedom. Its conductances and killing rates induce the Markov transitions established in Sec.~\ref{sec:multiscale-probability}. Consequently, the walk traverses the prefab through coarse port transitions without resolving fine-grid diffusion inside each cell.

When the coarse walk leaves the skeleton through an exterior-facing port $\beta$, we restore its fine-boundary exit location using the first-exit statistics of the final coarse cell $K_j$. Conditioned on its final transition $\alpha\to\beta$, the exit distribution is
\begin{equation}
\label{eq:hybrid-boundary-lifting}
Q_{\alpha\beta}^{K_j}(b)= - \left(\displaystyle\sum_{a\in\Gamma_{j,\alpha}}c_{ab}
+\sum_{x\in i_j}R^{K_j}_{x\alpha}F_{\kappa,j}(x,b)\right) / (S_{K_j})_{\alpha\beta}
\end{equation}
Here $\Gamma_{j,\alpha}$ and $\Gamma_{j,\beta}$ are ports of $K_j$, and $i_j$ denotes its interior nodes. The numerator resolves the effective transport reaching the particular boundary node $b$, while summing over $\Gamma_{j,\beta}$ recovers the coarse inter-port conductance. The sampled exit is transformed into world coordinates, and exterior WoS resumes.

\begin{algorithm}[t]
\caption{Hybrid continuous--multiscale Monte Carlo walk}
\label{alg:hybrid-walk}
\begin{algorithmic}[1]
\State Advance exterior WoS until a prefab trigger or $\partial\Omega$.
\State Find $K_j\ni x$ and sample $b\sim\widehat F_{\kappa,j}(x,\cdot)$, or absorb.
\State Map $b$ to its coarse port $\alpha$ on the internal skeleton.
\State Walk using $S_{\mathrm{skel}}$ until reaching an outer port or absorption.
\State Sample the final fine exit $b_{\mathrm{out}}\sim Q_{\alpha\beta}^{K_j}$.
\State Resume exterior WoS from $b_{\mathrm{out}}$ and back to step 1. 
\end{algorithmic}
\end{algorithm}

\paragraph{Interface consistency.}
The three stages operate at distinct spatial resolutions: the local first-exit kernel connects a continuous position to the coarse-cell boundary, the condensed operators transport the walk across the internal skeleton, and the conditional exit distribution restores the location on the prefab boundary. This avoids repeated fine-scale simulation while retaining the boundary resolution needed for continuous coupling. The inner-trigger/outer-exit separation prevents immediate repeated switching at artificial interfaces. Screening, steering, and temporal edits update the compiled kernels and condensed operators through the event-reweighting mechanisms developed above, without modifying the exterior WoS sampler.

\subsection{Adaptive Multiscale Transport}
\label{sec:adaptive-multiscale-transport}

For each query $q$, we adaptively select the resolution of each prefab $K$ based on its spatial proximity to the query. Specifically, $K$ uses the fine level if $d(q,K)\le r$ and the coarse level otherwise, where $d(q,K)$ denotes the Euclidean distance to the prefab and $r$ controls the refinement radius. This assignment is determined once per query and remains fixed across all subsequent random walks, allowing nearby regions to retain detailed transport while distant regions use cheaper approximations. Both levels follow the same prefab interface, so switching resolution requires no modification to the global solver.

\section{Experiments}
\label{sec:experiments}

We evaluate how accurately reusable transport prefabs preserve heterogeneous diffusion, how their responses change under physical edits, and how much work they eliminate during repeated Monte Carlo queries. Our main benchmark uses a heterogeneous medium divided into $27$ coarse cells, with stochastic first-exit kernels coupled by an internal Schur-complement skeleton. The complete Hybrid alternates between continuous WoS outside the medium, cached first exits inside coarse cells, skeleton transitions, and fine-resolution exit reconstruction. We use independent finite-volume references for solution accuracy and single-threaded CPU timings. Additional local and multi-instance experiments examine individual components and spatial reuse.

\subsection{Stochastic Compilation and Multiscale Accuracy}
\label{sec:exp-compile}

\paragraph{Stochastic basis and coarse transfer.}
We first ask whether estimating local first exits introduces noticeable errors into the condensed operator that drives the coarse walk. On a strongly heterogeneous conductivity field, individual multiscale basis values exhibit visible Monte Carlo fluctuations, whereas their port-aggregated transfer response is substantially more stable. Table~\ref{tab:exp-basis} shows that increasing the compilation budget reduces both errors, but the transport response converges to a small error long before every basis value becomes accurate. Aggregation over the fine-scale boundary interactions therefore filters much of the sampling noise relevant to the coarse walk. This is the reason to compile and reuse port-level transport rather than reconstruct a uniformly precise interior basis for each query.

\begin{table}[t]
\centering
\caption{Convergence of stochastic basis functions and aggregated coarse transport on a heterogeneous local medium.}
\label{tab:exp-basis}
\small\setlength{\tabcolsep}{4pt}
\begin{tabular}{rrrr}
\toprule
Paths/node & Basis error & Transfer error & Build (s) \\
\midrule
64 & 12.08\% & 0.330\% & 1.45 \\
256 & 6.10\% & 0.126\% & 5.57 \\
1024 & 3.03\% & 0.083\% & 23.09 \\
\bottomrule
\end{tabular}
\end{table}

\paragraph{Port resolution and entrance interpolation.}
The Hybrid introduces spatial compression at two stages: the Schur skeleton retains a finite set of boundary ports, while the entrance-dependent exit kernel is interpolated from cached anchors. We vary these resolutions independently and measure the resulting full-solution error (Fig.~\ref{fig:exp-spatial-ablation}). With only eight ports, the skeleton is too coarse to preserve the exit dependence of smooth boundary data; increasing the port count sharply improves the solution, with gains tapering toward the finest configuration. The oscillatory case is less sensitive to port count in this particular field. Entrance-kernel spacing exhibits a different pattern: sparse anchors mainly degrade the smooth-field result, while denser anchors recover the spatial variation in where a walk exits a coarse cell. These observations separate the roles of coarse inter-port transport and spatially resolved entrance-to-exit interpolation.

\begin{figure}[t]
    \centering
    \includegraphics[width=0.9\linewidth]{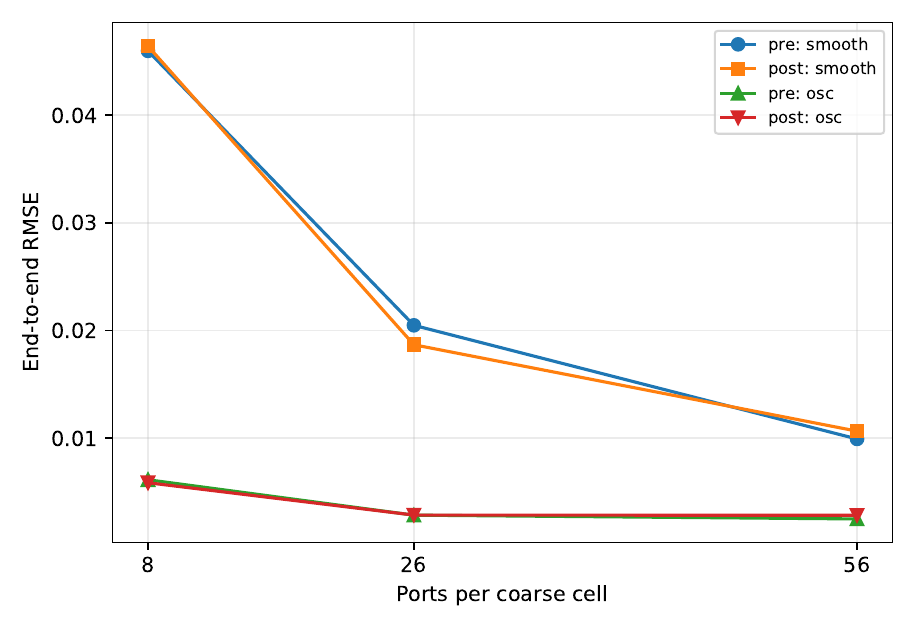}
    {\small Coarse-cell port count}
    \vspace{6pt}
    \includegraphics[width=0.9\linewidth]{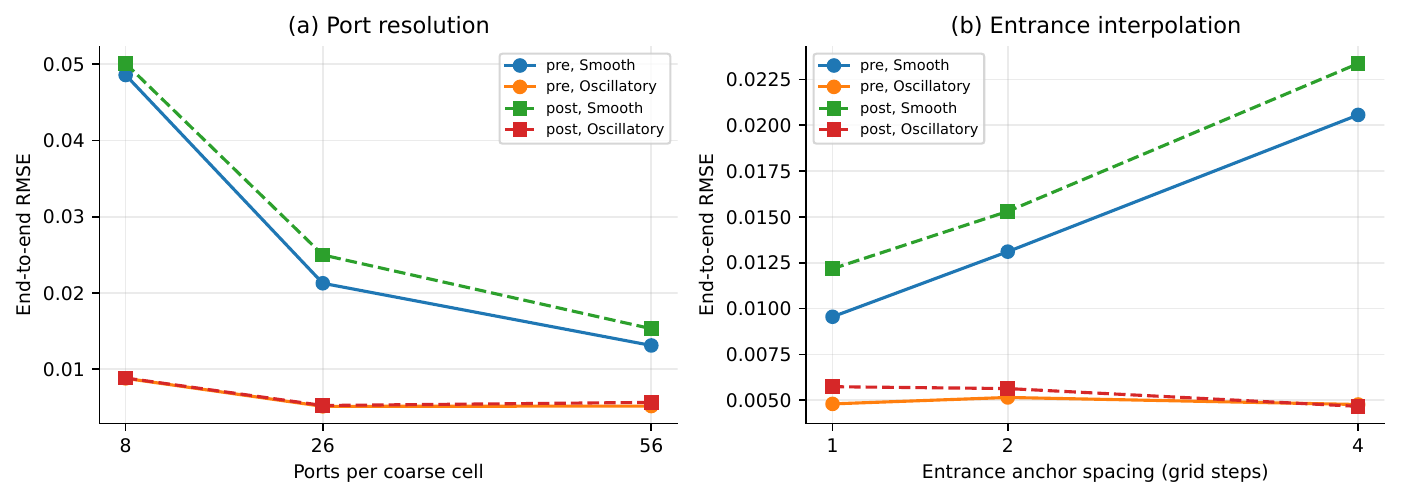}
    {\small Entrance interpolation spacing}

    \caption{Sensitivity of full-Hybrid solution error to port resolution and entrance-kernel interpolation spacing.}
    \label{fig:exp-spatial-ablation}
\end{figure}

\paragraph{Offline event-bank size.}
A larger event bank improves the estimates used to build the prefab but also increases its one-time compilation and storage costs. We therefore repeat the complete hybrid evaluation with independently compiled banks of increasing size while keeping the online sampling budget fixed (Fig.~\ref{fig:exp-banks}). The solution errors change only moderately and do not decrease monotonically across these finite-bank runs. This behavior indicates that, once the coarse transitions are reasonably resolved, online sampling and spatial approximation can dominate differences between event-bank sizes. The default bank balances reusable transport quality with compilation cost.

\begin{figure}[t]
\centering\includegraphics[width=\linewidth]{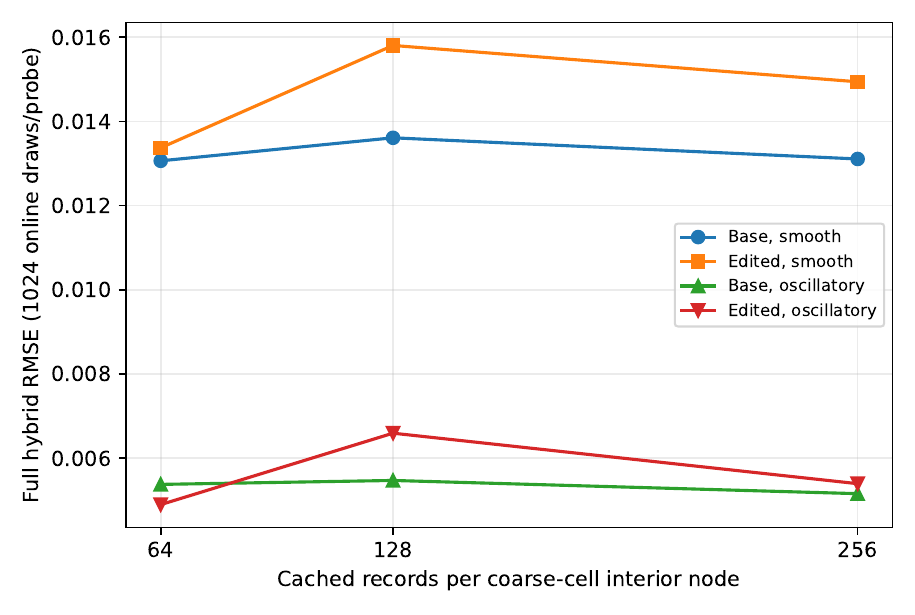}
\caption{Full-Hybrid solution error under different offline event-bank sizes at a fixed online budget.}
\label{fig:exp-banks}
\end{figure}

\subsection{End-to-End Solver Performance}
\label{sec:exp-queries}

We next compare the complete $F\!\to\!S\!\to\!$ fine-exit Hybrid with the official variable-coefficient WoS implementation of~\cite{sawhney2022variable}. Both solvers evaluate the same heterogeneous configuration before and after a screening edit, under smooth and oscillatory boundary conditions. The comparison covers the full query path rather than a standalone local exit sampler. We use the same number of online draws and assess both solutions against a common refined finite-volume reference.

The main difference is the work repeated during each query. Conventional WoS must revisit the heterogeneous coefficients as its paths evolve; the Hybrid replaces most of those interior steps with cached exits and coarse transitions. Consequently, the Hybrid achieves substantially lower online costs across all four configurations (Table~\ref{tab:exp-official}), with speedup ratios of approximately $22$--$28\times$ relative to the archived official WoS measurements. Its solution accuracy remains comparable to the direct solver: the Hybrid achieves lower RMSE in three of the four configurations, while direct WoS performs slightly better in the unedited smooth case. The online-budget comparison in Fig.~\ref{fig:exp-official} further illustrates the separation in computational cost as the number of sampled trajectories increases.

\begin{table*}[t]
\centering
\caption{End-to-end Hybrid and variable-coefficient WoS at matched online sample counts (4096 paths/query). Both methods are evaluated on the same machine against a common refined finite-volume reference. Time is for 12 queries; preprocessing is excluded.}
\label{tab:exp-official}
\small\setlength{\tabcolsep}{9pt}
\begin{tabular}{llrrrrr}
\toprule
& & \multicolumn{2}{c}{Hybrid $F\!\to\!S\!\to\!$exit} & \multicolumn{2}{c}{Official WoS~\cite{sawhney2022variable}} & Online ratio \\
\cmidrule(lr){3-4}\cmidrule(lr){5-6}
Medium & Boundary & RMSE & ms & RMSE & ms & Official/Hybrid \\
\midrule
Base & Smooth & 0.0099 & 137.5 & 0.0082 & 3078.9 & $22.4\times$ \\
Base & Oscillatory & 0.0025 & 114.2 & 0.0058 & 3130.3 & $27.4\times$ \\
Edited & Smooth & 0.0114 & 111.2 & 0.0124 & 3138.3 & $28.2\times$ \\
Edited & Oscillatory & 0.0030 & 113.5 & 0.0076 & 3043.8 & $26.8\times$ \\
\bottomrule
\end{tabular}
\end{table*}

\begin{figure}[t]
\centering
\includegraphics[width=\linewidth]{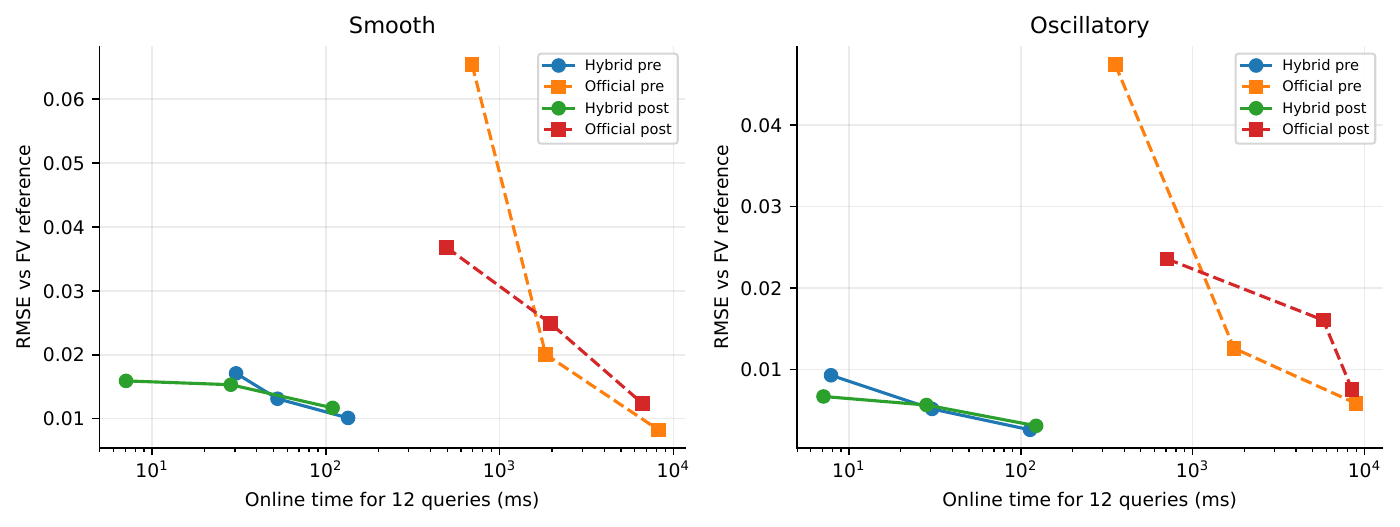}
\caption{Online runtime and solution error under increasing sample budgets for the two complete solvers.}
\label{fig:exp-official}
\end{figure}

\subsection{Incremental Editing Accuracy}
\label{sec:exp-increments}

\paragraph{Edit-local transport increments.}
The resolvent estimator in Sec.~4.2 updates the condensed transport operator from trajectories originating in the edited region. Its main challenge is not estimating a large operator response, but recovering the comparatively small change caused by an edit. We test screening and Laplace-domain capacity changes over 40 translated, rotated, and rescaled absorber poses, using independent discrete resolvent increments as references. Each pose is evaluated with four separately compiled event banks. Figure~\ref{fig:exp-pose-examples} shows representative absorber poses used in the replay and full-Hybrid edit experiments.

\begin{figure}[t]
\centering
\includegraphics[width=\linewidth]{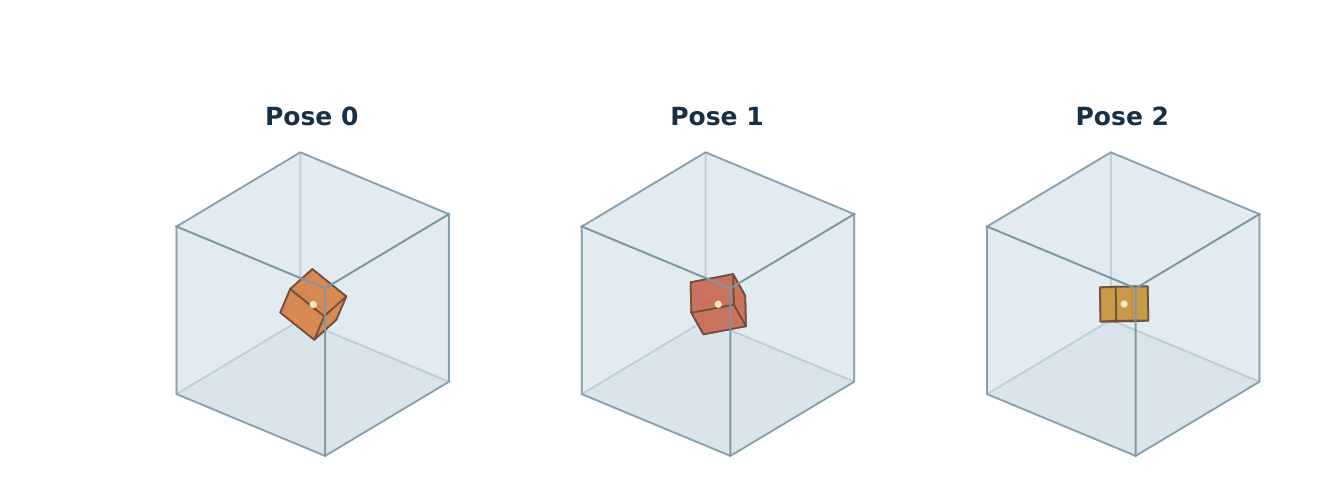}
\caption{Three representative three-dimensional absorber poses from the strong-edit experiments, shown as rotated and translated screening volumes within their host transport prefabs.}
\label{fig:exp-pose-examples}
\end{figure}

\begin{figure}[t]
\centering
\includegraphics[width=\linewidth]{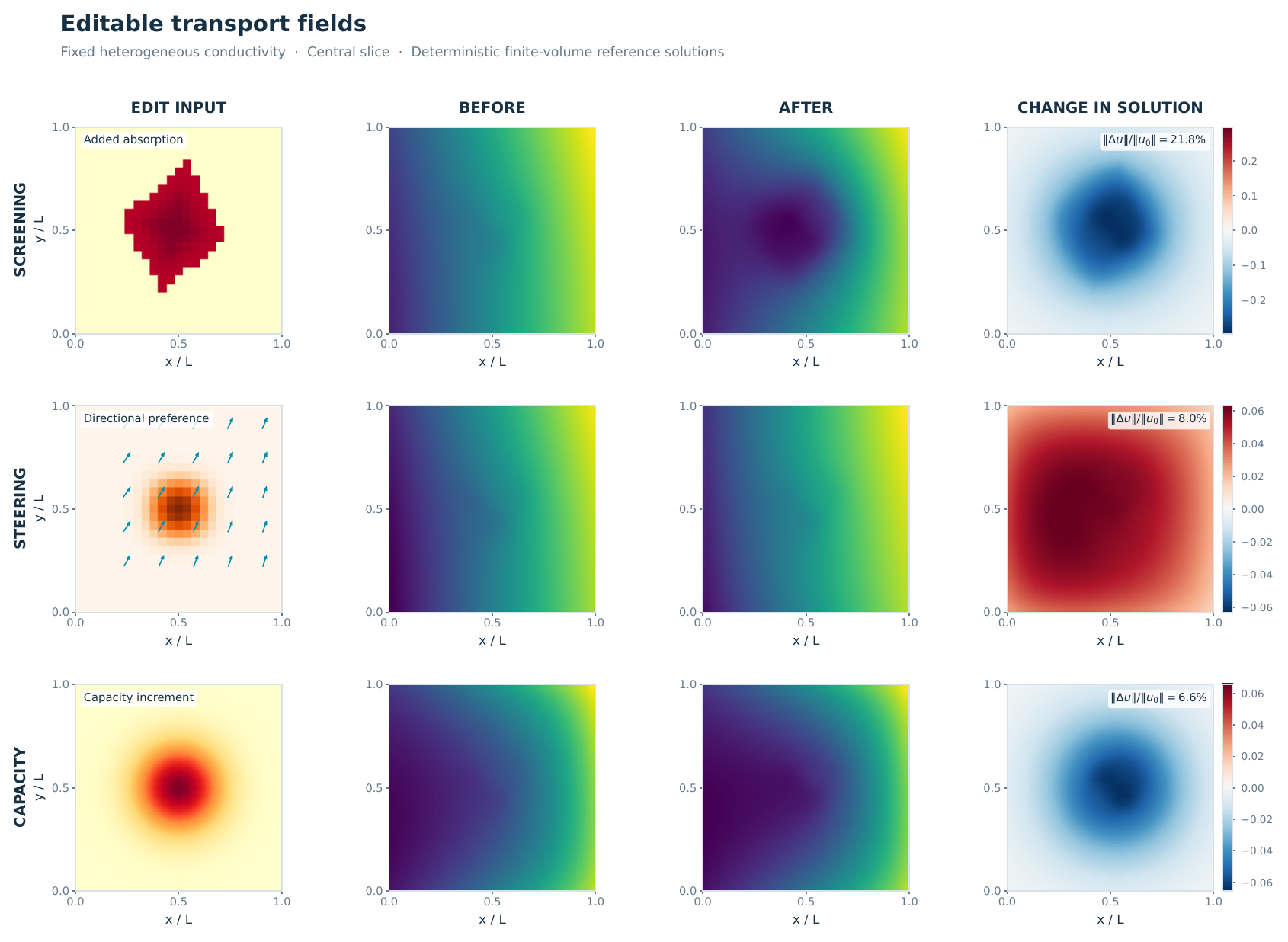}
\caption{Representative screening, harmonic steering, and Laplace-domain capacity edits. Each row shows the edit input, deterministic finite-volume reference solutions before and after editing, and the resulting solution difference.}
\label{fig:exp-edit-inputs}
\end{figure}

Figure~\ref{fig:exp-edit-inputs} summarizes the three edit controls evaluated in this section and makes explicit what is changed before the transport increments are estimated. The incremental estimates remain consistent across poses (Fig.~\ref{fig:exp-pose-errors}), with sub-percent average relative error for ordinary edit supports. When an edit occupies a much smaller region, the relative error rises: fewer starting nodes contribute, and the change being estimated is smaller. Screening and capacity edits follow almost the same pattern because both enter the local operator as nonnegative diagonal perturbations. Increasing the number of cached paths reduces the pose-level error (Fig.~\ref{fig:exp-replay-budget}), showing that the replay quality can be improved offline without regenerating trajectories for each edit.


\begin{figure}[t]
\centering\includegraphics[width=\linewidth]{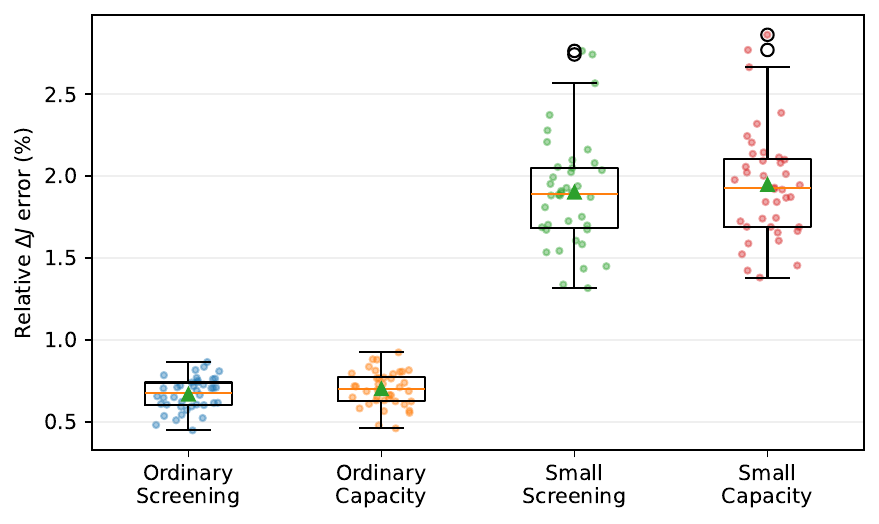}
\caption{Variation in incremental transfer error across absorber poses and edit-support sizes.}
\label{fig:exp-pose-errors}
\end{figure}

\begin{figure}[t]
\centering\includegraphics[width=\linewidth]{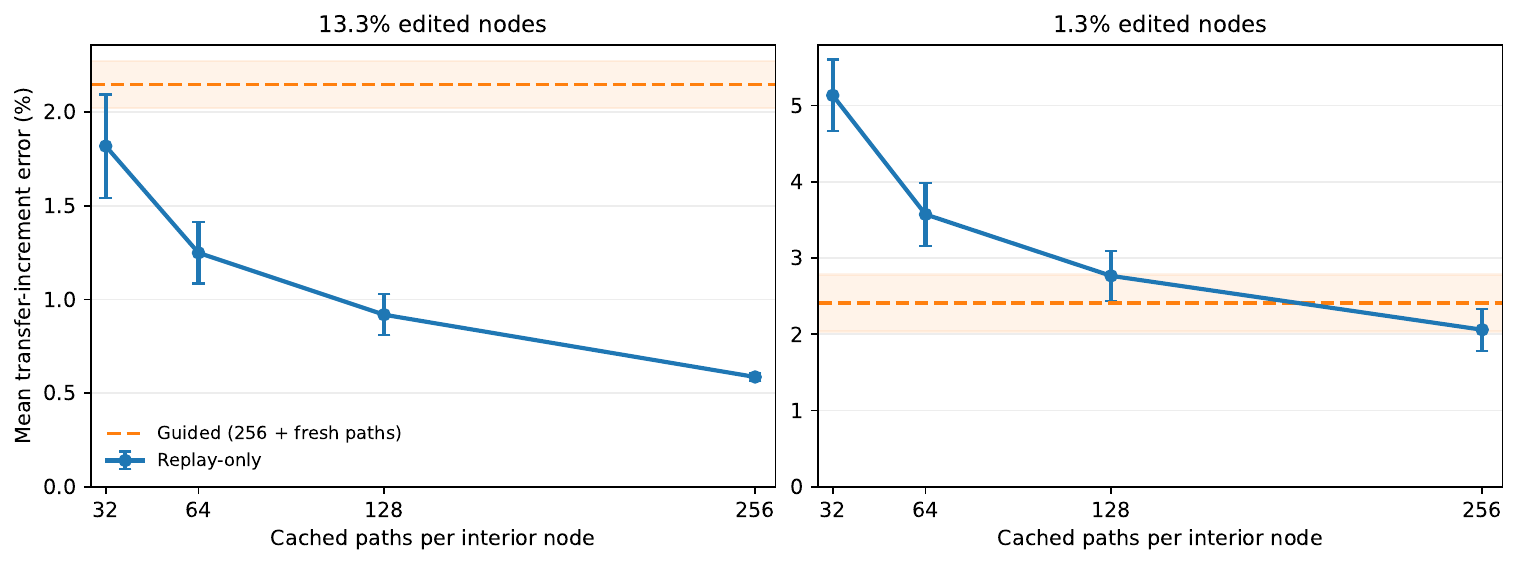}
\caption{Effect of the precompiled trajectory budget on edit-local transfer-increment accuracy.}
\label{fig:exp-replay-budget}
\end{figure}

\paragraph{Strong screening edits in the complete Hybrid.}
Accurate local operator increments must also produce meaningful changes in the assembled solution. We therefore apply 24 relatively strong screening edits of different shapes, positions, and orientations to the complete Hybrid, reusing a fixed event bank. These edits visibly alter the physical response rather than merely perturbing it near the Monte Carlo noise level. The estimated solutions track the overall changes in both boundary-condition families, although their incremental errors are more pronounced than in the isolated local-operator test (Table~\ref{tab:exp-strong-pose}). Figure~\ref{fig:exp-strong-pose} makes the distinction clear: a reliable local transfer update still passes through port compression, entrance interpolation, and random sampling when used in the full solver. The experiment therefore measures how those approximations accumulate during substantial material edits.

\begin{table}[t]
\centering
\caption{Full-Hybrid response accuracy across 24 strong screening edits after enforcing valid screening support.}
\label{tab:exp-strong-pose}
\small
\begin{tabular}{lrrr}
\toprule
Boundary & Change RMS & Final RMSE & $\Delta$ RMSE \\
\midrule
Oscillatory & 0.1074 & 0.0214 & 0.0211 \\
Smooth & 0.1472 & 0.0327 & 0.0281 \\
\bottomrule
\end{tabular}
\end{table}

\begin{figure}[t]
\centering
\includegraphics[width=\linewidth]{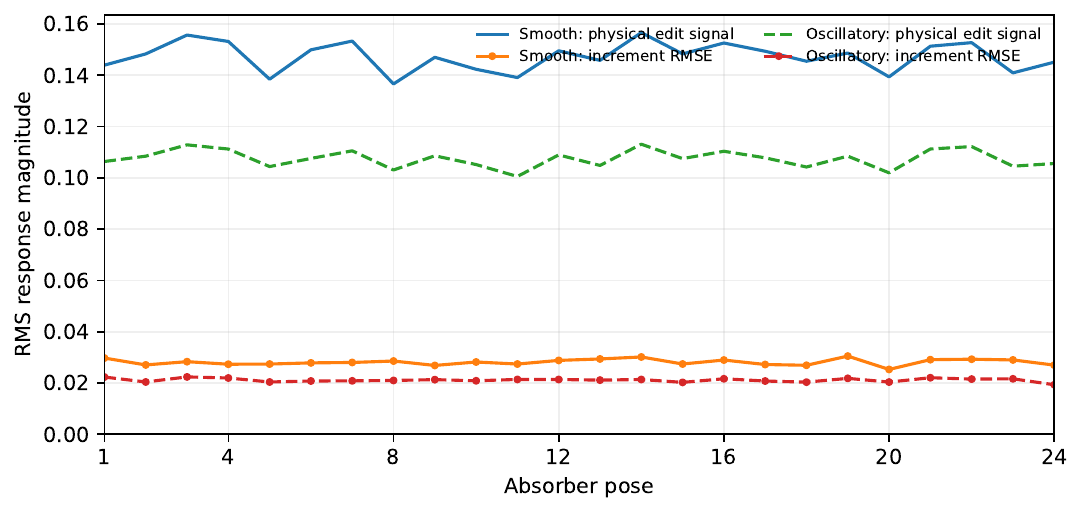}
\caption{Magnitude of the reference physical edit response and the corresponding full-Hybrid increment error.}
\label{fig:exp-strong-pose}
\end{figure}

\paragraph{Response to smaller edits.}
We also evaluate 48 more localized screening edits with two online sampling budgets (Fig.~\ref{fig:exp-48poses}). Here the final fields stay close to their references, but the physical edit signal is weaker and is harder to resolve separately from sampling noise. Increasing the online draws improves both final-solution and increment errors, particularly for the oscillatory boundary condition. Together with the strong-edit sweep, these results show why edit fidelity should be inspected through both the complete post-edit field and its change from the unedited solution: the two can behave differently for localized perturbations.

\begin{figure}[t]
\centering\includegraphics[width=\linewidth]{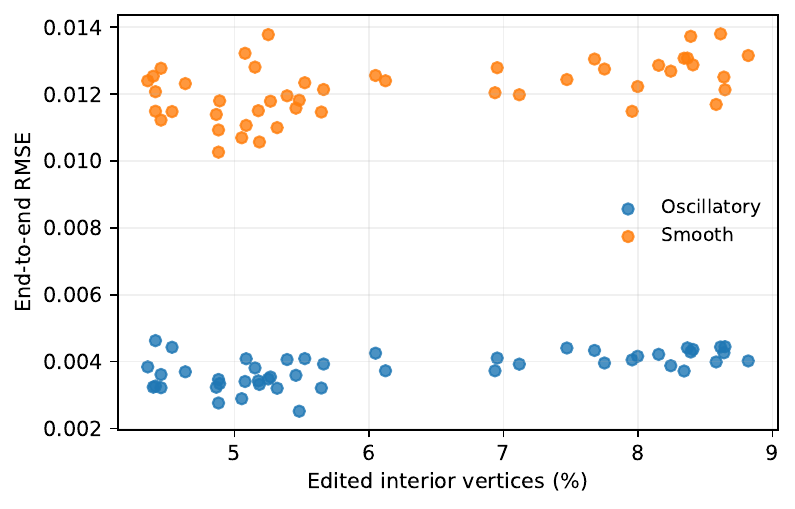}
\caption{Solution and incremental-response errors under 48 localized screening edits.}
\label{fig:exp-48poses}
\end{figure}

\subsection{Edit-Local Schur Updates}
\label{sec:exp-edits}

Repeated screening edits require updating the condensed transport operator. We compare our edit-local U-statistic with full Schur refactorization, with both methods reusing the same geometry, conductances, port mappings, and pre-edit operator. The local method reconstructs only the contributions of edited interior nodes from cached event records; the baseline refactorizes the interior matrices and recomputes the Schur complements of all $27$ coarse cells. We measure their operator-update costs for progressively larger edited regions, using $256$ cached records per interior node (Table~\ref{tab:exp-cost-detail} and Fig.~\ref{fig:exp-edit-cost}).

The local update is particularly inexpensive for small, spatially confined edits. Its cost grows as more edited nodes contribute to the resolvent correction, while full Schur reconstruction remains relatively stable because it refactorizes the same coarse-cell system regardless of edit support. Consequently, the greatest speedups occur for local edits, reaching $25.9$--$42.9\times$ below $5\%$ edited nodes. The advantage decreases gradually with support size but remains $1.8\times$ even when the entire interior is edited. This behavior illustrates the main benefit of caching transport: editing a small region requires updating its influence rather than reconstructing the complete operator.

\begin{table}[t]
\centering
\caption{Schur-operator update time for varying edit-support sizes (ms; median of five runs, 256 cached records/node).}
\label{tab:exp-cost-detail}
\small\setlength{\tabcolsep}{4pt}
\begin{tabular}{rrrr}
\toprule
Edited nodes & Local $\Delta S$ & Full $S$ & Speedup \\
\midrule
0.8\% & 5.2 & 186.8 & 36.2$\times$ \\
2.2\% & 5.0 & 215.7 & 42.9$\times$ \\
4.7\% & 7.9 & 205.1 & 25.9$\times$ \\
14.1\% & 35.0 & 163.1 & 4.7$\times$ \\
21.6\% & 42.1 & 176.4 & 4.2$\times$ \\
31.4\% & 55.2 & 170.7 & 3.1$\times$ \\
100.0\% & 91.8 & 163.9 & 1.8$\times$ \\
\bottomrule
\end{tabular}
\end{table}

\begin{figure}[t]
\centering
\includegraphics[width=\linewidth]{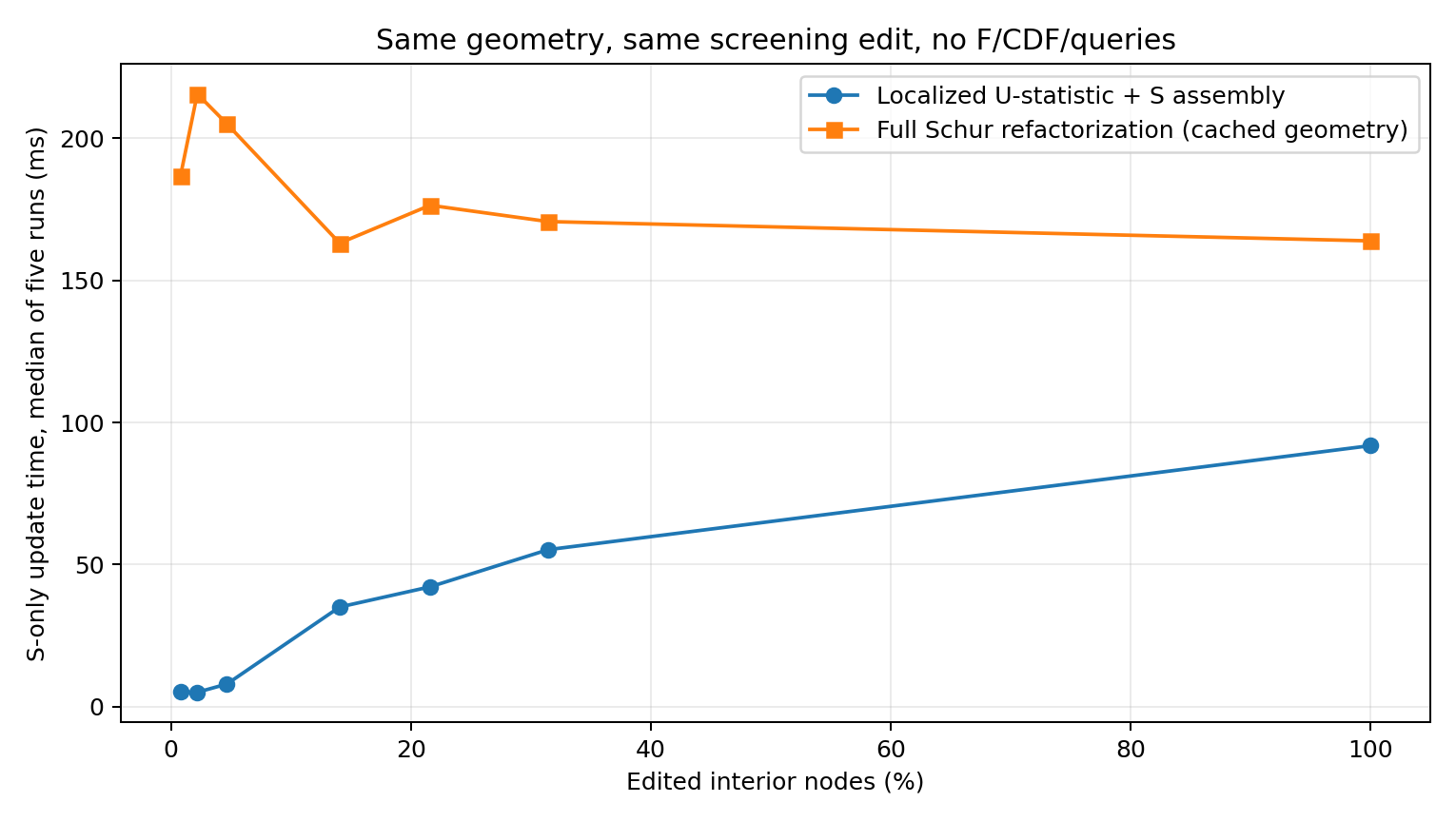}
\caption{Local and full Schur-update costs as the edited region expands.}
\label{fig:exp-edit-cost}
\end{figure}

\paragraph{Accuracy--cost trade-off.}
The local correction is estimated from finite cached paths, so we also measure how its approximation affects the coarse response. We evaluate each updated operator with the same exact post-edit interface coupling and compare its coarse solution with the exact edited Schur response. Increasing the event-record budget from $64$ to $256$ generally moves the local-update results toward lower response error (Fig.~\ref{fig:exp-edit-accuracy}). Across edit sizes, the approximate updates retain small coarse-response errors while requiring substantially less work than full refactorization. The curve exposes a controllable exchange between the number of reused samples and the fidelity of the edited transport operator.

\begin{figure}[t]
\centering
\includegraphics[width=\linewidth]{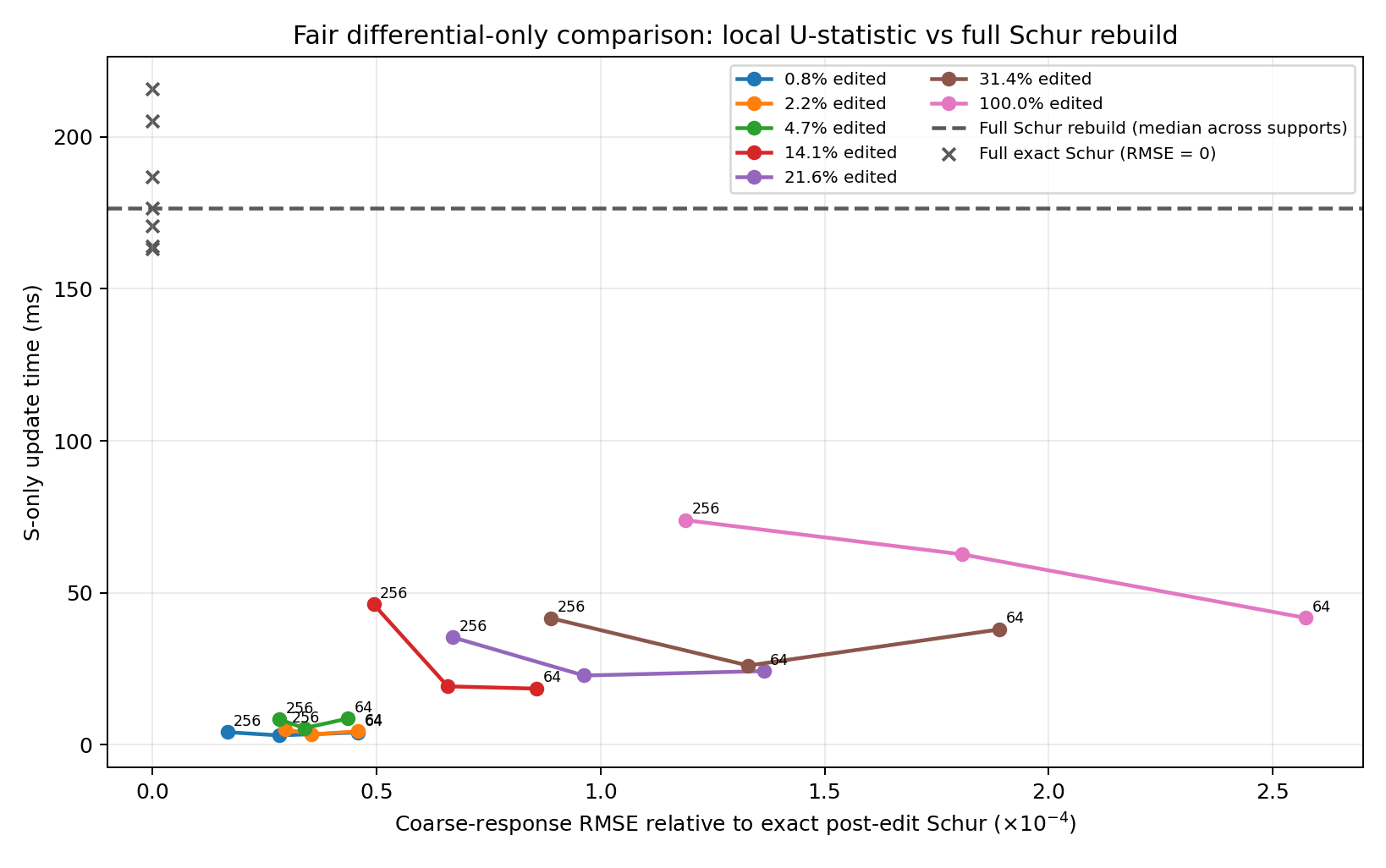}
\caption{Coarse-response error and Schur-update time for 64--256 cached records/node. Full refactorization gives the exact edited Schur operator.}
\label{fig:exp-edit-accuracy}
\end{figure}

\subsection{Transport Steering and Temporal Responses}
\label{sec:exp-controls}

\paragraph{Harmonic steering.}
To test whether the cached transport can be directed without changing the diffusion field, we assign positive preferences to the prefab ports and apply the harmonic endpoint reweighting of Sec.~4.3. The edited solution increasingly favors the selected exit directions while retaining the structure imposed by the heterogeneous medium. Against directly solved discrete reference problems, the estimate follows this change as steering becomes stronger (Fig.~\ref{fig:exp-controls}a). The discrepancy grows with steering strength because the same finite set of cached exits represents a more unevenly weighted distribution. No new spatial trajectories are needed to explore these steering choices.

\paragraph{Heat-capacity editing in the Laplace domain.}
A capacity change modifies the time scale of diffusion rather than the direction of its stationary transport. After Laplace transformation, the capacity field appears as an additional effective screening term, so the existing marked paths can be reweighted for different transform parameters. Figure~\ref{fig:exp-controls}b shows that the estimates remain close to the independently solved discrete responses throughout the tested range. Unlike steering, the relative error changes only slightly as the effective capacity contribution increases. This demonstrates reuse of the same local transport samples for a second, physically distinct edit control.

\begin{figure*}[t]
\centering
\begin{minipage}[t]{0.49\textwidth}\centering
\includegraphics[width=\linewidth]{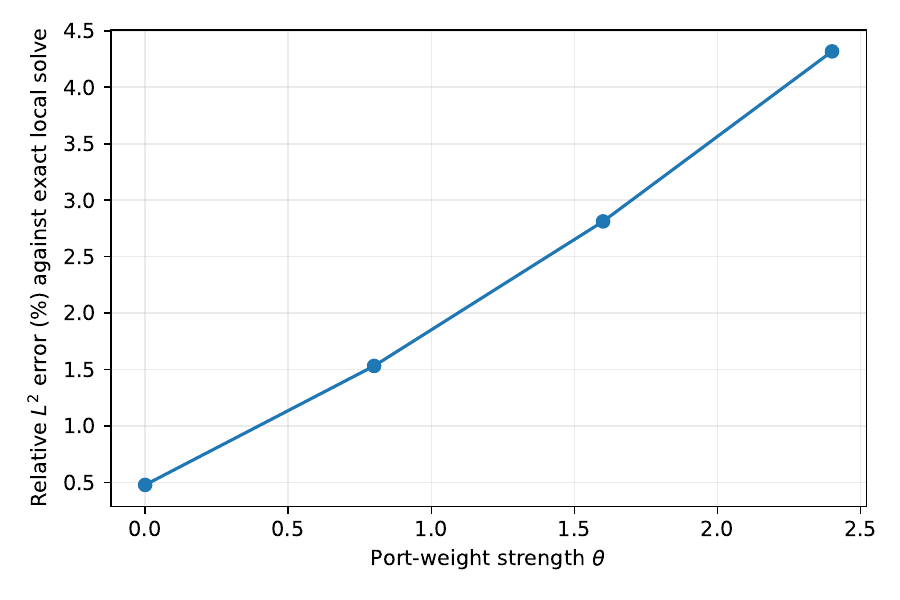}\\[-2pt]
{\small (a) Harmonic gradient steering}
\end{minipage}\hfill
\begin{minipage}[t]{0.49\textwidth}\centering
\includegraphics[width=\linewidth]{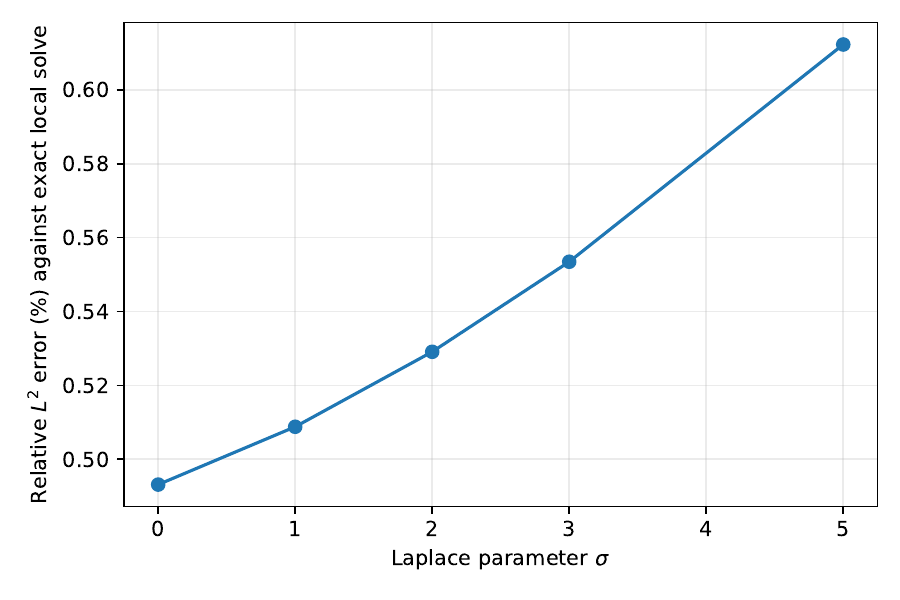}\\[-2pt]
{\small (b) Laplace-domain heat-capacity edits}
\end{minipage}
\caption{Accuracy of cached harmonic steering and Laplace-domain heat-capacity updates against direct discrete solutions.}
\label{fig:exp-controls}
\end{figure*}

\paragraph{Direct-time heat-ball transport.}
We additionally evaluate a direct time-domain implementation using WoHB. With a nonzero sinusoidal initial condition, the unit-cube heat equation has an analytic solution, allowing the effects of capacity and screening to be checked separately. The heat-ball estimates follow the expected decay across the tested capacities; increasing the capacity slows the evolution at a fixed physical time. Reweighting the same walks for different screening values produces estimates close to fresh direct-time sampling (Table~\ref{tab:exp-wohB}, Fig.~\ref{fig:exp-wohB}). Thus Poisson events also provide reusable screening information along time-dependent walks.

\begin{table}[t]
\centering
\caption{Direct heat-ball estimation and mark replay with nonzero initial data.}
\label{tab:exp-wohB}
\small\setlength{\tabcolsep}{6pt}
\begin{tabular}{rrr}
\toprule
Capacity $m$ & Direct WoHB MAE & Mark replay MAE \\
\midrule
1 & 0.0018 & 0.0014 \\
1.7 & 0.0065 & 0.0066 \\
3 & 0.0038 & 0.0037 \\
\bottomrule
\end{tabular}
\end{table}

\begin{figure}[t]
\centering\includegraphics[width=\linewidth]{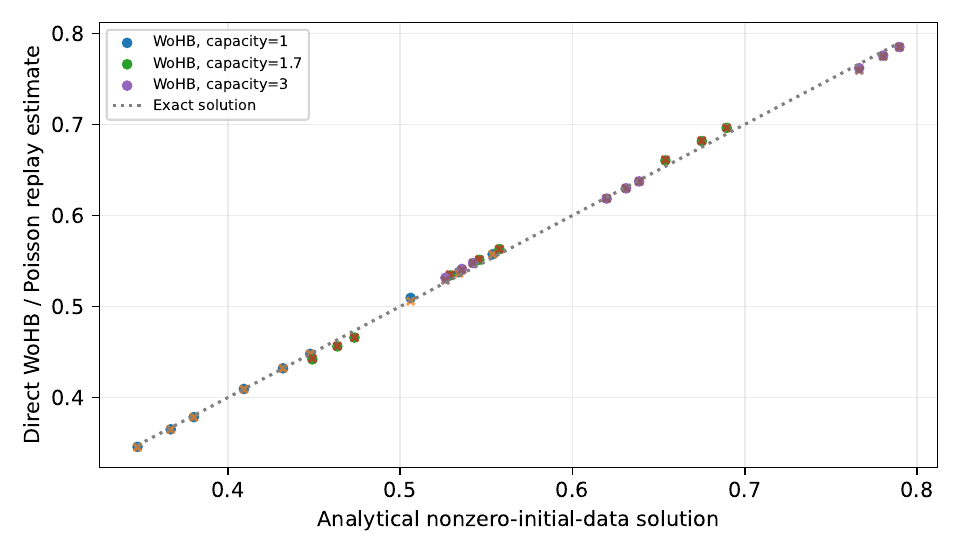}
\caption{Direct heat-ball and mark-replay estimates compared with the analytic transient solution.}
\label{fig:exp-wohB}
\end{figure}

\paragraph{Spatially varying capacity and cached clocks.}
We next consider a discrete conductance graph whose thermal capacity increases in a localized region. In this test, the cache also stores the visited states and holding times needed to reconstruct elapsed physical time. Replaying these histories under several capacity fields reproduces the corresponding matrix-exponential reference responses with similar accuracy (Fig.~\ref{fig:exp-capacity-clock}). The estimates capture the delayed propagation caused by increased local storage, demonstrating that spatial transport trajectories can be reused even when their physical clocks change.

\begin{figure}[t]
\centering\includegraphics[width=\linewidth]{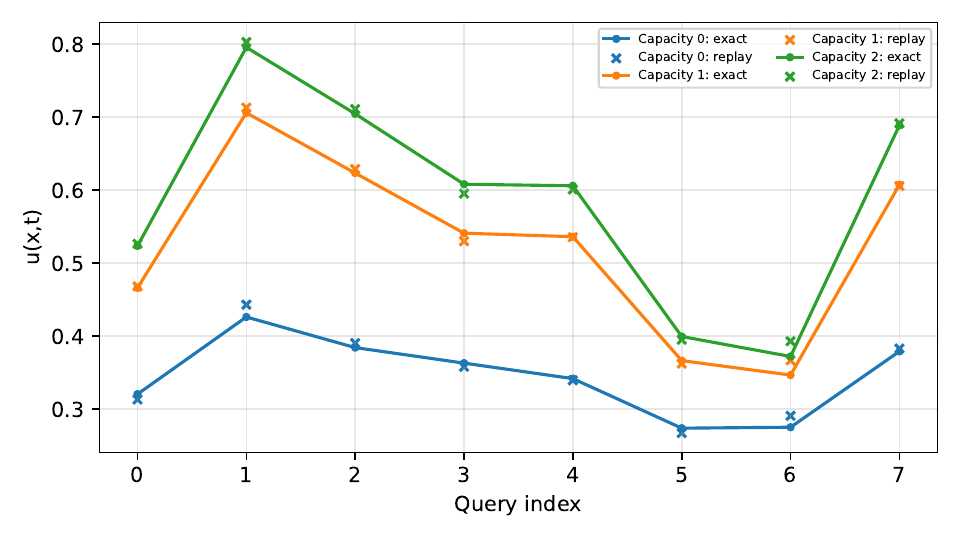}
\caption{Replayed physical-time clocks on a conductance graph compared with the transient discrete reference.}
\label{fig:exp-capacity-clock}
\end{figure}

\subsection{Spatial Reuse and Query-Centric Resolution}
\label{sec:exp-query-lod}
To evaluate query-centric resolution selection at scene scale, we randomly distribute 12 spatially separated instances of the same heterogeneous medium in a homogeneous exterior (Fig.~\ref{fig:exp-lod-scene}). Each instance reuses two precompiled transport representations, comprising $6^3$ and $24^3$ coarse cells, respectively. Both resolutions follow the same continuous--discrete hybrid walk, coupling local first-exit kernels, Schur-complement skeleton transport, and fine-boundary exit reconstruction. For each query $q$, an instance $K$ uses the fine representation if $d(q,K)\le r$ and the coarse representation otherwise, with the assignment fixed throughout all subsequent walks. We vary $r$ across successive query-to-prefab distance thresholds, progressively increasing the number of fine-resolution instances, and compare the resulting responses with an all-fine reference (Fig.~\ref{fig:query-lod-tradeoff}).

\begin{figure}[t]
    \centering
    \includegraphics[width=\linewidth]{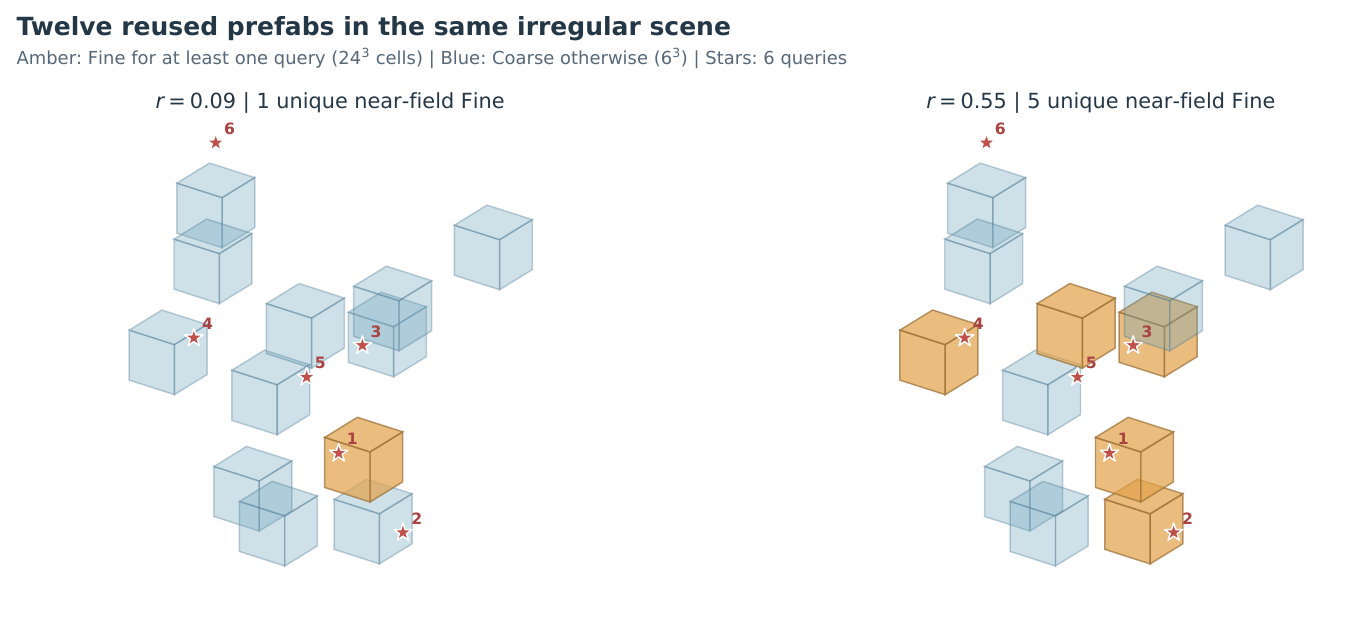}
    \caption{Spatial arrangement of 12 reusable transport prefabs in a homogeneous exterior. Six query locations induce different proximity-based resolution assignments, with nearby instances evaluated using the fine representation.}
    \label{fig:exp-lod-scene}
\end{figure}

As the refinement radius increases, the paired discrepancy between the adaptive LoD response and the all-fine response decreases rapidly as nearby prefabs switch to finer representations. Further refinement yields diminishing improvements, while the measured online CPU time generally increases, with minor timing fluctuations (Fig.~\ref{fig:query-lod-tradeoff}). Both smooth and oscillatory boundary conditions exhibit this behavior. The paired discrepancy isolates the effect of resolution selection under matched random sampling; it should not be interpreted as the absolute error relative to the continuum PDE solution. These results demonstrate a controllable trade-off between online computational cost and transport representation fidelity while reusing the same heterogeneous medium across instances.

\begin{figure}[t]
\centering
\includegraphics[width=\linewidth]{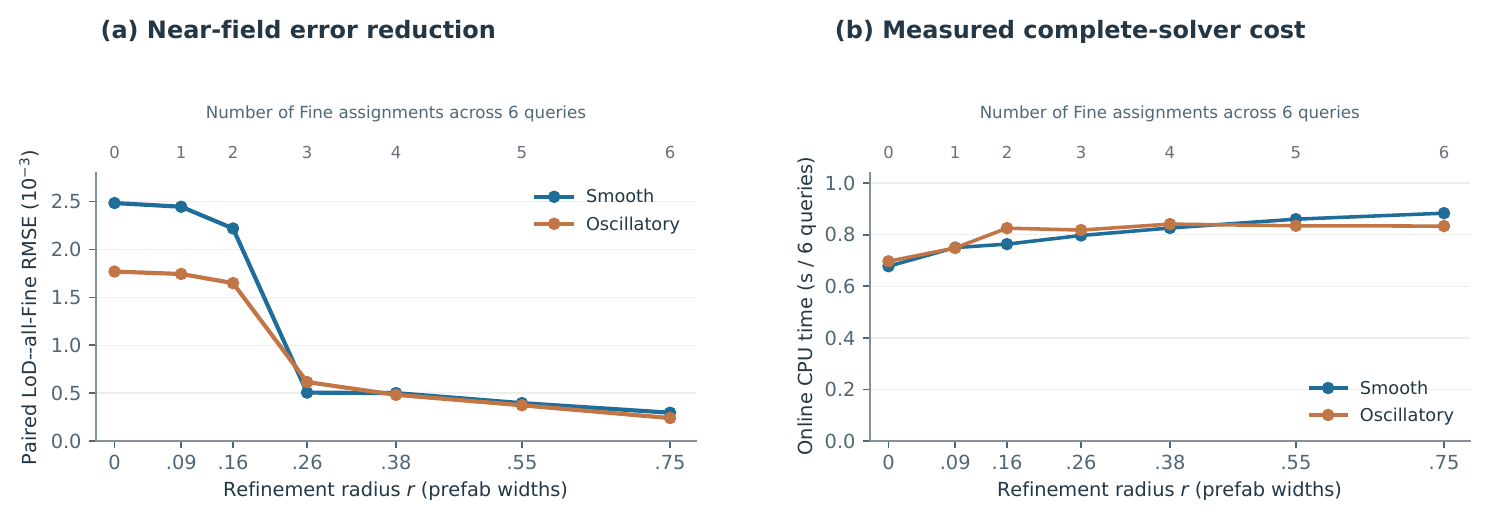}
\caption{Query-centric resolution selection across twelve irregularly placed prefabs. Left: paired LoD--all-Fine response discrepancy as nearby instances switch to finer precompiled representations. Right: measured complete-solver CPU time. The paired discrepancy measures representation differences rather than absolute PDE solution error.}
\label{fig:query-lod-tradeoff}
\end{figure}
\section{Conclusion}
We demonstrated that fine-scale heterogeneous transport can be compiled into reusable stochastic prefabs, enabling efficient multiscale Monte Carlo simulation and interactive material editing. By combining probabilistic operator condensation with cached-path replay, our framework substantially accelerates repeated queries and localized updates while supporting spatial composition, transport steering, and transient diffusion. The main limitation is that the compiled representation assumes fixed underlying diffusivity up to scaling and requires upfront preprocessing, making it most beneficial for repeatedly queried or edited scenes. Extending this framework to dynamically changing diffusion structures remains an important direction for future work.
\bibliographystyle{ACM-Reference-Format}
\bibliography{references}

@String{Springer = "Springer-Verlag" }

@book{bensoussan1978asymptotic,
  author = {Bensoussan, Alain and Lions, Jacques-Louis and Papanicolaou, George},
  title = {Asymptotic Analysis for Periodic Structures},
  publisher = {North-Holland},
  year = {1978},
}

@article{allaire1992homogenization,
  author = {Allaire, Gr{\'e}goire},
  title = {Homogenization and Two-Scale Convergence},
  journal = {SIAM Journal on Mathematical Analysis},
  volume = {23}, number = {6}, pages = {1482--1518}, year = {1992},
}

@book{jikov1994homogenization,
  author = {Jikov, Vladimir V. and Kozlov, Sergey M. and Oleinik, Olga A.},
  title = {Homogenization of Differential Operators and Integral Functionals},
  publisher = {Springer}, year = {1994},
}

@article{hou1997msfem,
  author = {Hou, Thomas Y. and Wu, Xiao-Hui},
  title = {A Multiscale Finite Element Method for Elliptic Problems in Composite Materials and Porous Media},
  journal = {Journal of Computational Physics},
  volume = {134}, number = {1}, pages = {169--189}, year = {1997},
}

@article{e2003hmm,
  author = {E, Weinan and Engquist, Bjorn and Huang, Zhongyi},
  title = {Heterogeneous Multiscale Method: A General Methodology for Multiscale Modeling},
  journal = {Physical Review B}, volume = {67}, pages = {092101}, year = {2003},
}

@article{efendiev2013gmsfem,
  author = {Efendiev, Yalchin and Galvis, Juan and Hou, Thomas Y.},
  title = {Generalized Multiscale Finite Element Methods ({GMsFEM})},
  journal = {Journal of Computational Physics},
  volume = {251}, pages = {116--135}, year = {2013},
  doi = {10.1016/j.jcp.2013.04.045}
}

@article{malqvist2014localization,
  author = {M{\aa}lqvist, Axel and Peterseim, Daniel},
  title = {Localization of Elliptic Multiscale Problems},
  journal = {Mathematics of Computation},
  volume = {83}, number = {290}, pages = {2583--2603}, year = {2014},
}

@article{owhadi2017multigrid,
  author = {Owhadi, Houman},
  title = {Multigrid with Rough Coefficients and Multiresolution Operator Decomposition from Hierarchical Information Games},
  journal = {SIAM Review}, volume = {59}, number = {1}, pages = {99--149},
  year = {2017},
}

@article{muller1956continuous,
  author = {Muller, Mervin E.},
  title = {Some Continuous {Monte Carlo} Methods for the {Dirichlet} Problem},
  journal = {The Annals of Mathematical Statistics},
  volume = {27}, number = {3}, pages = {569--589}, year = {1956},
}

@article{sawhney2020montecarlo,
  author = {Sawhney, Rohan and Crane, Keenan},
  title = {Monte Carlo Geometry Processing: A Grid-Free Approach to {PDE}-Based Methods on Volumetric Domains},
  journal = {ACM Transactions on Graphics},
  volume = {39}, number = {4}, articleno = {123}, year = {2020},
}

@article{sawhney2022variable,
  author = {Sawhney, Rohan and Seyb, Dario and Jarosz, Wojciech and Crane, Keenan},
  title = {Grid-Free Monte Carlo for {PDEs} with Spatially Varying Coefficients},
  journal = {ACM Transactions on Graphics},
  volume = {41}, number = {4}, year = {2022},
}

@article{sawhney2023stars,
  author = {Sawhney, Rohan and Miller, Bailey and Gkioulekas, Ioannis and Crane, Keenan},
  title = {Walk on Stars: A Grid-Free Monte Carlo Method for {PDEs} with Neumann Boundary Conditions},
  journal = {ACM Transactions on Graphics},
  volume = {42}, number = {4}, articleno = {80}, year = {2023},
}

@article{miller2023boundary,
  author = {Miller, Bailey and Sawhney, Rohan and Crane, Keenan and Gkioulekas, Ioannis},
  title = {Boundary Value Caching for Walk on Spheres},
  journal = {ACM Transactions on Graphics},
  volume = {42}, number = {4}, articleno = {82}, year = {2023},
}

@article{huang2026probes,
  author = {Huang, Wanchao and Zhu, Yutian and Fang, Qing and Liu, Ligang},
  title = {Probe-Based Walk on Spheres for Efficient Path Reusing},
  journal = {ACM Transactions on Graphics},
  volume = {45}, number = {4}, articleno = {129}, year = {2026},
}

@article{jambon2026decomposed,
  author = {Jambon, Cl{\'e}ment and Nabizadeh, Mohammad Sina and Konakovi{\'c} Lukovi{\'c}, Mina},
  title = {Walk on Decomposed Subdomains: A Hybrid {Monte Carlo}--Deterministic Solver for Elliptic {PDEs}},
  journal = {ACM Transactions on Graphics},
  volume = {45}, number = {4}, articleno = {132}, year = {2026},
}

@misc{maesumi2026steklov,
  author = {Maesumi, Arman and Makadia, Tanish and Anderson, Aruna and Phongpanangam, Oras and Solomon, Justin and Ritchie, Daniel},
  title = {Monte Carlo Steklov Operators for Large-Scale Geometry Processing in the Wild},
  year = {2026}, eprint = {2606.05581}, archivePrefix = {arXiv},
  url = {https://arxiv.org/abs/2606.05581}
}

@misc{cimpean2026extended,
  author = {C{\^i}mpean, Iulian and Grecu, Andreea and Zarnescu, Arghir},
  title = {Extended Walk-on-Spheres Algorithm for Linear and Nonlinear Elliptic Problems of Divergence-Type},
  year = {2026}, eprint = {2608.07017}, archivePrefix = {arXiv},
  url = {https://arxiv.org/abs/2608.07017}
}

@misc{bao2026heatstars,
  author = {Bao, Anchang and Shen, Enya and Zhang, Yongjun and Liu, Zhongwei and Wang, Jianmin},
  title = {Walking on Heat Stars for Parabolic Heat Equations with Neumann Boundary Conditions},
  year = {2026}, eprint = {2606.16578}, archivePrefix = {arXiv},
  url = {https://arxiv.org/abs/2606.16578}
}

@misc{zhou2026timedependent,
  author = {Zhou, Zihong and Sawhney, Rohan and d'Eon, Eugene and Jarosz, Wojciech},
  title = {Grid-Free Monte Carlo for Time-Dependent Diffusion},
  year = {2026}, eprint = {2609.12306}, archivePrefix = {arXiv},
  url = {https://arxiv.org/abs/2609.12306}
}

@article{deaconu2018heatballs,
  author  = {Deaconu, Madalina and Herrmann, Samuel},
  title   = {Initial-boundary value problem for the heat equation---A stochastic algorithm},
  journal = {The Annals of Applied Probability},
  volume  = {28},
  number  = {3},
  pages   = {1943--1976},
  year    = {2018},
}
\appendix
\end{document}